\documentclass[11pt]{article}

\usepackage{amsmath}
\usepackage{amsfonts}
\usepackage{amssymb}         
\usepackage{amsopn}
\usepackage{theorem}          
\usepackage{latexsym}
\usepackage{graphicx}        
\usepackage{mathrsfs}		
\usepackage{psfrag}
\usepackage{color}
\usepackage{verbatim}
\usepackage{url}
\usepackage{multirow}
\usepackage{setspace}
\usepackage{booktabs,longtable,rotating}
\usepackage{afterpage}
\usepackage{sidecap}
\usepackage{marvosym}
\usepackage[caption=false]{subfig}
\usepackage{enumitem}
\usepackage{algorithm}
\usepackage[noend]{algorithmic}
\usepackage{float}
\usepackage{times}
\usepackage{relsize}
\usepackage{lipsum}
\usepackage{geometry}
\newcommand{\Sig}{\Sigma}
\newcommand{\vp}{\varphi}
\newcommand{\eps}{\epsilon}
\newcommand{\ra}{{\rightarrow}}

\newcommand{\N}{\ensuremath{\mathbb{N}}}
\newcommand{\Q}{\ensuremath{\mathbb{Q}}}
\newcommand{\Z}{\ensuremath{\mathbb{Z}}}
\newcommand{\R}{\ensuremath{\mathbb{R}}}
\renewcommand{\S}{\ensuremath{\mathcal{S}}}
\newcommand{\A}{\ensuremath{\mathcal{A}}}
\newcommand{\D}{\ensuremath{\mathcal{D}}}
\newcommand{\G}{\ensuremath{\mathcal{G}}}
\newcommand*\diff{\mathop{}\!\mathrm{d}}
\newcommand{\vx}{\ensuremath{\vec{x}}}
\newcommand{\vD}{\ensuremath{\vec{D}}}
\newcommand{\vI}{\ensuremath{\vec{I}}}

\newcommand{\ceil}[1]{{\left\lceil{#1}\right\rceil}}

\newtheorem{theorem}{Theorem}[section]

\newtheorem{corollary}[theorem]{Corollary}

\newtheorem{proposition}[theorem]{Proposition}
\newtheorem{definition}[theorem]{Definition}
\newtheorem{condition}{Condition}

\newenvironment{proof}
 {{\sl Proof.}\hspace*{1 ex}}%
 {{\nopagebreak\hspace*{\fill}$\Box$\par\vspace{12pt}}}

\title{Toward breaking the curse of dimensionality:\\ an FPTAS for
  stochastic dynamic programs \\with multidimensional actions and scalar
  states}

\author{%
Nir Halman
\thanks{Jerusalem School of Business Administration, The
    Hebrew University, Jerusalem, Israel, {\tt halman@huji.ac.il}}
\and Giacomo Nannicini
\thanks{IBM T.~J.~Watson, Yorktown Heights, NY, {\tt nannicini@us.ibm.com}}
} 

\begin{document}

\maketitle

\begin{abstract}
  We propose a Fully Polynomial-Time Approximation Scheme (FPTAS) for
  stochastic dynamic programs with multidimensional action, scalar
  state, convex costs and linear state transition function. The action
  spaces are polyhedral and described by parametric linear
  programs. This type of problems finds applications in the area of
  optimal planning under uncertainty, and can be thought of as the
  problem of optimally managing a single non-discrete resource over a
  finite time horizon. We show that under a value oracle model for the
  cost functions this result for one-dimensional state space is ``best
  possible'', because a similar dynamic programming model with
  two-dimensional state space does not admit a PTAS.

  The FPTAS relies on the solution of polynomial-sized linear programs
  to recursively compute an approximation of the value function at
  each stage. Our paper enlarges the class of dynamic programs that
  admit an FPTAS by showing, under suitable conditions, how to deal
  with multidimensional action spaces and with vectors of continuous
  random variables with bounded support. These results bring us one
  step closer to overcoming the curse of dimensionality of dynamic
  programming.
\end{abstract}



\section{Introduction}
\label{s:intro}
A dynamic program (DP) is a mathematical model for sequential decision
making. DPs are widely used by the operations research community to
model and solve a large variety of problems concerning optimal
planning under uncertainty. Unfortunately, DPs are affected by the
{\em curse of dimensionality} -- an expression coined by Richard
E.~Bellman more than 50 years ago~\cite{Be61} -- that makes their
solution very difficult in practice. There is a large body of work
devoted to ways of circumventing the curse, possibly foregoing optimality
or approximation guarantees: this is discussed in
Section~\ref{s:literature}.

This paper deals with a class of discrete-time finite-horizon
stochastic DPs characterized by a scalar state and a multidimensional
action, where the optimal action at each stage and state can be
computed as the solution of a linear program (LP). We now give a more
detailed description of the underlying mathematical model. The
evolution of the state of the DP is governed by the transition
function $f_t$ and the equation $I_{t+1} = f_t(I_t, \vx_t, \vD_t), \;
t=1,\dots,T$, where: $t$ is the discrete time index, $I_t \in \S_t$ is
the state of the system at time $t$ ($\S_t$ is the {\em state space}
at stage $t$), $\vx_t \in \A_t(I_t)$ is the action or decision to be
selected at time $t$ after observing state $I_t$ ($\A_t(I_t)$ is the
{\em action space} at stage $t$ and state $I_t$), $\vD_t$ is a vector
of interstage independent random variables (r.v.s) over the sample
space $\D_t$, and $T$ is the number of time periods. The random vector
$\vD_t$ represents an exogenous information flow, and can be
continuous or discrete. The cost function $g_t(I_t, \vx_t, \vD_t)$
gives the cost of performing action $\vx_t$ from state $I_t$ at time
$t$ for each possible realization of $\vD_t$. The function $g_{T+1}$
is also called the ``terminal cost function'', and it gives the cost
of leaving the system in state $I_{T+1}$ at the end of the time
horizon under consideration. Costs are accumulated over all time
periods: the total incurred cost is equal to $\sum_{t = 1}^T g_t(I_t,
\vx_t, \vD_t) + g_{T+1}(I_{T+1})$. We use $\vx_t$, $\vD_t$ to
emphasize vector quantities, whereas the other quantities are
scalars. The problem is that of choosing policies $\pi_t : \S_t \to
\bigcup_{I_t \in \S_t} \A(I_t)$ in order to determine a sequence of
actions $\vx_1 \in \A(I_1),\dots,\vx_T \in \A(I_T)$ that minimizes the
expectation of the total incurred cost. This problem is called a {\em
  stochastic dynamic program}. Formally, we want to determine:
\begin{align}
  z^\ast(I_1) = \min_{\pi_1,\dots,\pi_T}
  \mathbb{E} \left[\sum_{t = 1}^T g_t(I_t, \pi_t(I_t), \vD_t) +
    g_{T+1}(I_{T+1}) \right], \tag{DP} \label{eq:dp}\\
  \text{subject to: } I_{t+1} = f_t(I_t, \pi_t(I_t), \vD_t), \qquad
  \forall t=1,\dots,T, \notag
\end{align}
where $I_1$ is the initial state of the system and the expectation is
taken with respect to the joint probability distribution of the random
variables $\vD_t$. The seminal result of Bellman \cite{bellmandp} is
that the solution to this problem is given by a recursion:
\begin{theorem} \label{thm:Bellman} For every
  initial state $I_1$, the optimal value $z^\ast(I_1)$ of the DP is
  given by $z_1(I_1)$, where $z_1$ is the function defined by
  $z_{T+1}(I_{T+1})=g_{T+1}(I_{T+1})$ together with the recursion:
  \begin{equation}
  \label{eq:vf}
  z_t(I_t)=\min_{\vx_t \in \A_t(I_t)} \mathbb{E} \left\{
  g_t(I_t,\vx_t,\vD_t)+z_{t+1}(f_t(I_t,\vx_t,\vD_t)) \right\}, \quad
  t=1,\ldots,T.
  \end{equation}
\end{theorem}
When the state and action spaces are finite, and the expectations can
be computed with a finite process, this recursion gives a finite
algorithm to compute the optimal value. However, this algorithm may
require exponential time in general. If the state and action spaces
are not finite (e.g.\ when they are intervals or polyhedra as is
discussed in the present paper), or if we cannot compute expectations
in finite time (e.g.\ when there are continuous r.v.s that do not
possess a closed-form formula for the expectation), computing
$z^*(I_1)$ may be intractable.

This paper proposes an FPTAS for a specific class of DPs. A
polynomial-time approximation scheme (PTAS) is an approximation scheme
that returns a solution whose cost is at most $(1+\epsilon)$ times the
optimal cost, where $\epsilon > 0$ is a given parameter. A PTAS runs
in time polynomial in the (binary) input size. If the algorithm also
runs in time polynomial in $1/\epsilon$, we call it an FPTAS. To
construct an FPTAS, we must impose additional structure on the DP,
because it is known that not all DPs admit a polynomial-time
approximation scheme (see e.g.~\cite[Thm.~9.2]{halman14full}).  We
assume that the state $I_t$ is a scalar, whereas the action $\vx_t$
can be vector-valued. We show in Sect.~\ref{s:hardness} that if the
cost functions are only accessible via value oracles, the restriction
that $I_t$ is a scalar cannot be relaxed. At each stage, the vector
$\vD_t$ is composed of a given number of independent but not
necessarily identically distributed r.v.s with bounded support, that
can have continuous or discrete distribution. The vectors $\vD_t$ are
also interstage independent. (If the vectors are not independent, it
is known that the DP may be APX-hard and therefore cannot admit an
FPTAS, see e.g.\ \cite[Thm.~10.1, Cor.~10.2]{halman14full}.)
Furthermore, we assume that: the transition functions $f_t$ are affine;
the costs functions $g_t$ decompose to the sum of a deterministic
piecewise linear convex function and a convex function that may depend
on the r.v.s via an affine transformation; and the action spaces
$\A_t(I_t)$ are polyhedral sets, described as the feasible
region of a parametric LP with right-hand side parameter $I_t$. A
formal description of the assumptions is given in
Sect.~\ref{s:assumptions}. DPs satisfying these assumptions can be
seen as multistage stochastic LPs, see e.g.~\cite{birgesp}, with a
single variable linking consecutive stages (the state variable) and
stochastic r.h.s.\ vector. The constraint matrix and the matrix
linking the state between consecutive stages (also called ``technology
matrix'') are deterministic. However, unlike the sample average
approximation approach typically used in multistage stochastic LP, our
algorithm is deterministic and its running time depends polynomially
on the number of stages. Our framework allows some degree of convex
nonlinearity in the objective function, see
Condition~\ref{con:functions}(ii). One can think of the type of
optimization problems addressed here as stochastic resource management
problems with a single resource, see Sect.~\ref{s:application}. An
example of a problem that exhibits the required characteristics in the
context of energy resource allocation is described in
\cite{powell12smart}. More applications are described in
\cite{nascimento13}: management of water in a reservoir, of cement for
a construction company, of the cash reserve for an investment bank.

\paragraph{Organization of the paper.} This paper is organized as
follows. In Sect.~\ref{s:preliminaries} we define our notation, state
our assumptions, and present an example of an application. In
Sect.~\ref{s:hardness} we discuss the necessity of our assumptions and
show some related hardness results. Sect.~\ref{s:size} introduces an
algorithm to deal with the issue of number sizes growing too
rapidly. Sect.~\ref{s:exp} contains an FPTAS for the sum of r.v.s, and
shows how to efficiently compute the expectation in the DP in the
presence of continuous r.v.s. Sect.~\ref{s:scheme} puts all the
building blocks together to design the FPTAS for problem \eqref{eq:dp}
satisfying our assumptions. Sect.~\ref{s:conclusions} concludes the
paper. In the rest of the Introduction, we position our paper as
compared to the existing literature, we provide a summary of our
contributions, and give an overview of the techniques used.

\subsection{Relevance to existing literature}
\label{s:literature}
Dynamic programming is an invaluable tool for sequential decision
making under uncertainty, and it has received a large amount of
attention. DPs are often very difficult to solve in practice; for this
reason, several approximate solution methodologies have been
developed. \cite{powelldp} discusses explicitly the three {\em curses
  of dimensionality} of DPs, namely: the large dimensionality of the
state space, of the action space, and the difficulty or impossibility
of computing expectations.

To deal with these curses, \cite{CDJ14} studies fixed-dimensional
stochastic dynamic programs with discrete state and action spaces over
a finite horizon. Assuming that the cost-to-go functions are discrete
$L^\natural$-convex (classes of discrete convex functions are
discussed later in this subsection), \cite{CDJ14} proposes a
pseudopolynomial-time approximation scheme that satisfies an arbitrary
pre-specified additive error guarantee. The proposed approximation
algorithm is a generalization of the explicit enumeration
algorithm. The main differences between our paper and \cite{CDJ14}
are: (i) \cite{CDJ14} considers discrete state and action spaces (as
opposed to continuous); (ii) it considers fixed dimensional (as
opposed to one-dimensional) state spaces; (iii) it gives additive (as
opposed to relative) error approximation; (iv) the running time of the
approximation algorithm in \cite{CDJ14} is pseudopolynomial (as
opposed to polynomial) in the binary size of the input. Both
\cite{CDJ14} and the current paper are based on generalization of the
technique of $K$-approximation sets and functions. Another relevant
work in dealing with the curses of dimensionality in options pricing
and optimal stopping is \cite{CG18}. While \cite{CG18} provides
rigorous guarantees on the approximation error (additive and, in some
cases, relative), the assumptions and techniques are very different
from our paper.


The discipline known as Approximate Dynamic Programming (ADP)
strives to deal with the three curses of dimensionality, while
hopefully providing theoretical guarantees on the solution
quality. Comprehensive references in this area are
\cite{bertsekasdp,powelldp}. In most cases, ADP approaches are content
with proving asymptotic convergence to a ``best possible'' value
function, e.g., the best value function approximation that can be
obtained from a given set of basis functions \cite{defarias03}. This
is considerably different from the approach presented in this paper:
we compute an $\epsilon$-approximate solution in polynomial time, for
any given $\epsilon > 0$.

Our paper gives a framework to construct an FPTAS for a continuous
optimization problem. The literature contains only a few general
frameworks that yield approximation schemes for non-discrete
optimization problems. The framework in \cite{swamy12} deals with
stochastic LPs, and is a randomized scheme, whereas the algorithm
given in the present paper is deterministic. The framework of
\cite{halman14full} deals with stochastic discrete DPs in which the
r.v.s are described explicitly as lists of scenarios $(d, \Pr (D=d))$,
and the single period cost functions possess either monotone or convex
structure. \cite{nannifptas15} studies a subclass of the DP model of
\cite{halman14full}, in which the single period cost functions are
assumed to possess convex structure, and provides a faster FPTAS from
both theoretical worst-case upper bounds and practical standpoint. An
extension of \cite{halman14full} to continuous state and action spaces
is given in \cite{nannifptascdpfull}: however,
\cite{nannifptascdpfull} still deals with scalar state and action
spaces, and discrete (scalar) r.v.s. Our paper lifts some of those
restrictions.

To construct an approximation scheme for \eqref{eq:dp}, we study the
problem of approximating the cumulative distribution function (CDF) of
$\sum_{i=1}^n X_i$ where $X_1,\dots,X_n$ are given r.v.s. This problem
is well known to be \#P-hard, see e.g.~\cite{kleinberg97allocating},
\cite[Thm.~4.1]{halman09}. It plays a central role in stochastic
optimization because many problems in this area inherit \#P-hardness
from it, e.g.~\cite{halman09,halman14full}. The problem is related to
counting knapsack solutions: given a binary knapsack $\sum_{i=1}^n a_i
x_i \le b$, a possible way to determine the number of feasible
solutions is to define discrete r.v.s $X_i$ equal to $a_i$ or $0$ with
probability $\frac{1}{2}$ each, and then compute $\Pr(\sum_{i=1}^n X_i
\le b)$. Based on the technique for counting knapsack solutions
developed in \cite{gopalan11fptas} (see also \cite{halman16}),
\cite{li14convolution} gives the first FPTAS to approximate the CDF of
the sum of r.v.s. The approach presented in this paper to deal with
vectors of r.v.s gives, as a by-product, an FPTAS for the same
problem, under weaker assumptions: the difference between
\cite{li14convolution} and our paper is discussed toward the end of
the next section.

\begin{figure}[h]
  \begin{center}
    \includegraphics[width=0.6\textwidth]{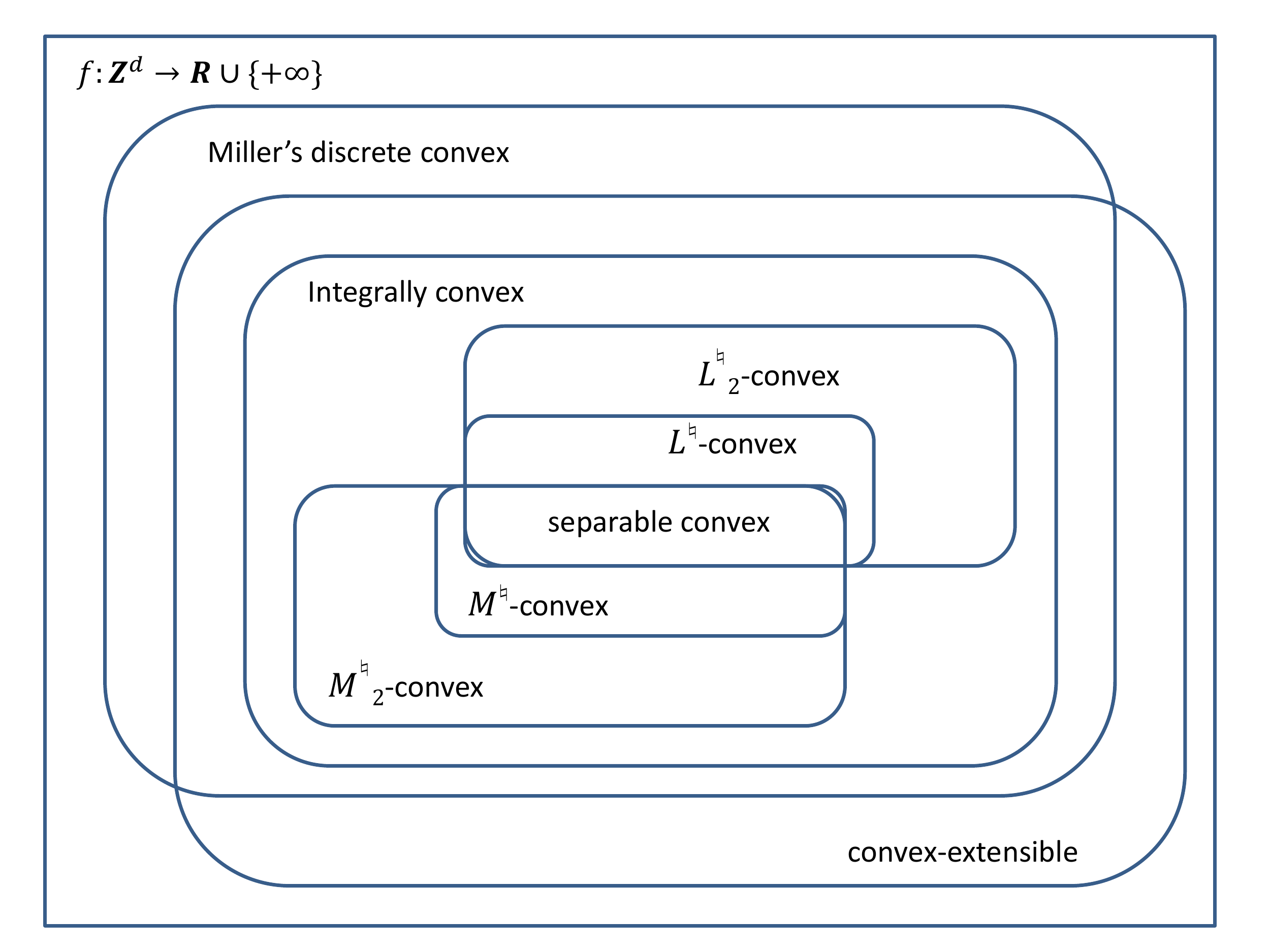}
  \end{center}
  \vskip-0.5cm
  \caption{Classes of discrete convex functions (adapted from Figure 1.15 in \cite{Mu03})} \label{fig:murota}
\end{figure}

Another area of relevance to the present paper is that of discrete
convexity, a discrete analog of convexity studied in the discrete
optimization literature. The first investigation of a class of
discrete functions for which local optimality implies global
optimality is due to Miller \cite{Mi71}.  Favati and Tardella
\cite{FT90} considered a certain special way of extending functions
defined over the integer lattice to piecewise-linear functions defined
over the real space, and they introduced the concept of ``integrally
convex functions''.  \cite{Mu03} introduced the concepts of ``convex
extensible'', ``$L^{\natural}$-convexity'' and
``$M^{\natural}$-convexity,'' in which ``L'' stands for ``Lattice''
and ``M'' stands for ``Matroid''; see \cite[Sec.~1.4]{Mu03}. The
relationship between several classes of discrete convex functions is
depicted in Fig.~\ref{fig:murota}. \cite[Sec.~3]{H15} shows that
separable convex functions in $\Z^d$ admit polylogarithmic-space
approximations, while any approximation of a two-dimensional
nondecreasing convex function in the sense of Miller may require
exponential space in the size of the domain of the function. Of the
classes represented in Fig.~\ref{fig:murota}, \cite{H15} leaves open
the approximability status of the class of convex-extensible
functions; it is settled in the present paper, see
Theorem~\ref{thm:bivarexthard} below.

\subsection{Our contribution}
Our paper makes progress toward overcoming some of the curses of
dimensionality, while keeping the strong theoretical guarantee of
finding an $\epsilon$-approximate solution. First, the framework
presented in this paper allows an action vector of arbitrary
dimension, and the running time depends polynomially on such
dimension. Second, we show how to handle vectors of continuous r.v.s
with bounded support under suitable conditions using an oracle for
their CDF, yielding an efficient approximate computation of expected
values of the cost functions. We remark that continuous r.v.s are
harder to handle than their discrete counterpart, because an exact
computation of expected values cannot be carried out in finite time in
general.

Thus, we summarize below our contribution regarding the three curses
of dimensionality:
\begin{enumerate}
\item We can handle scalar state spaces in polynomial time, as was
  first shown in \cite{nannifptascdpfull} under different
  assumptions. We show that if the cost functions are only accessible
  via value oracles, our continuous DP model with two-dimensional
  state space does not admit a PTAS.
\item We can handle vector action spaces of arbitrary dimension in
  polynomial time.
\item We can handle vectors of continuous r.v.s with bounded support
  in polynomial time.  We can also handle vectors of discrete r.v.s,
  extending the work of \cite{nannifptascdpfull} on scalar discrete
  r.v.s to the vector case. A more precise definition of the required
  conditions is given in Sect.~\ref{s:assumptions}.
\end{enumerate}


To prove that a DP with multidimensional state space does not admit a
PTAS when cost functions are described by an oracle, we prove in
Section~\ref{sec:multihard} the following result (we
define discrete convex-extensible in Section~\ref{sec:multihard}).
\begin{theorem}
  \label{thm:bivarexthard}
  There exists a family of two-dimensional positive nondecreasing
  discrete convex-extensible functions for which any approximation
  (with approximation ratio below a certain threshold) requires
  exponential space in the size of the domain of the function,
  regardless of the scheme used to represent the approximation.
\end{theorem}
This result is of independent interest and resolves the
approximability status of discrete convex-extensible functions, see
Fig.~\ref{fig:murota}. Since discrete convex-extensible functions
coincide with their convex envelope by definition, the above result
implies that two-dimensional nondecreasing piecewise-linear convex
functions do not admit an approximation that is polylogarithmic in the
size of the domain (i.e., polynomial in the (binary) size of the
domain's description).

To handle vectors of r.v.s, we assume that such vectors appear in the
cost and transition functions only via affine transformations. Our
FPTAS to approximate the weighted sum of r.v.s has weaker
assumptions than \cite{li14convolution}. The work of
\cite{li14convolution} assumes that we are given integer discrete
r.v.s and an exact oracle for the corresponding CDFs. It constructs an
FPTAS for the sum of r.v.s by discretizing a DP recurrence
relation. We relax their assumptions: the approach described in this
paper only requires an approximate oracle for the CDF of the r.v.s,
and it can also handle non-integer discrete r.v.s and continuous r.v.s
with bounded support. The running time is essentially the same. This
contribution of our paper extends the applicability of approximation
schemes for the sum of r.v.s.

The main constructive result shown in this paper is that a stochastic
DP satisfying some structural assumptions (Conditions
\ref{con:sets}-\ref{con:functions}) admits an FPTAS.  The
approximation scheme is remarkably short and simple: its pseudocode
consists of less than ten lines that use three different
subroutines. The majority of this paper is devoted to their
analysis. While the problem that we solve is continuous in nature, the
technique used to build the approximation is discrete in nature, and
may be useful for the construction of approximation schemes for
continuous optimization problems. We remark that our approximation
scheme is fully deterministic and does not rely on any form of random
sampling, unlike known sample average approximation schemes for
multistage optimization problems.

While the work presented in this paper is still far from the
generality of ADP approaches, there are several real-world
applications that fit into our framework, see
Sect.~\ref{s:application}. Furthermore, our framework provides much
stronger theoretical guarantees than ADP: despite the fact that an
exact solution to the problem may be impossible to compute in finite
time, we determine an $\epsilon$-approximation of the optimum for
arbitrary $\epsilon$. To the best of our knowledge, the present paper
is the first that uses FPTASs to tackle two curses of dimensionality
of DP, i.e., multidimensional action spaces and r.v.s: other FPTASs
for DPs mentioned in our literature review do not address these two
aspects.

\subsection{Techniques used}
\label{s:techniques}
The approximation scheme discussed in this paper relies on the concept
of $K$-approximation sets and functions, introduced in
\cite{halman09}. Intuitively, a $K$-approximation function (with $K =
1 + \epsilon > 1$) is a function that has a relative error of at most
$K-1$ (i.e., $\epsilon$) with respect to a given function. A
$K$-approximation set is a set of points that induces a
$K$-approximation function via linear interpolation. A
$K$-approximation set of a function with codomain $[1, U]$ can be
constructed with $O(\log_K U)$ points because we can allow the
function values to increase by roughly a factor $K$ between two
consecutive points. The papers \cite{halman09,halman14full} define
$K$-approximation sets for functions over discrete domains. The paper
\cite{nannifptascdpfull} extends this to continuous domains, but uses
both relative and absolute error measures. Since we work with
continuous domain, the definition of $K$-approximation sets used in
this paper is essentially the same as in \cite{nannifptascdpfull} when
only relative error is allowed, i.e., the additive error is fixed to
0. Formal definitions are given in Sect.~\ref{s:definitions}.

The approach used in this paper to deal with vector action spaces is
based on LP, and for this reason we impose structure on the problem
that allows us to use an exact LP solution methodology as a
subroutine. Since our approximation scheme proceeds recursively, we
must be careful to ensure that the size of the LPs does not grow too
fast. This is achieved with a scaling technique, finding a
$K$-approximation set with integer domain and codomain values for a
suitably scaled version of the function to be approximated; see
Sect.~\ref{s:size}.

Regarding the approximation for the weighted sum of r.v.s, our
construction is not based on counting knapsack solutions as in
\cite{li14convolution}, but rather on constructing compact
representations for the value function of a DP.  As a result, we
obtain a compact representation for the CDF of the sum of r.v.s over
the entire domain, rather than its value at a single point. The
construction is iterative and amounts to building a $K$-approximation
set for the CDF of a sum of r.v.s, adding one r.v.\ at a time. This is
discussed in Sect.~\ref{s:exp}.

\section{Preliminaries}
\label{s:preliminaries}
In this section, we formally introduce our notation and give the
necessary definitions, discussing the type of problems that can be
handled by the algorithm we propose.

\subsection{Definitions}
\label{s:definitions}
Let $\N, \Z, \Q, \R^+$ and $\R$ be the sets of nonnegative integers,
integers, rational numbers, nonnegative real numbers and real numbers
respectively. We distinguish between three types of continuous random
variables: continuous r.v.s that may have unbounded support and for
which $\Pr(X=x)=0$ for all $x$, mixed continuous r.v.s that may have
$\Pr(X=x)>0$ for some $x$, and truncated continuous r.v.s that have
bounded support with strictly positive probability on the endpoints
only.  Vectors are denoted with an overhead arrow, such as in
$\vx$. Subscripts should generally be intended to index a stage,
e.g.\ $\vx_t$ is the action vector at stage $t$ rather than the $t$-th
component of a vector $\vx$. When it is necessary to index a component
of a vector, it will be made explicit, e.g.\ $\vD_{t,i}$ will be the
$i$-th component of the vector $\vD_t$.  For $\ell, u \in \Z$, we
denote a finite set of the form $\{\ell, \ell + 1, \dots, u\}$ by
$[\ell,\dots,u]$, whereas $[\ell, u]$ denotes a real interval.

Given a real-valued function over a bounded linearly-ordered domain
$D$, $\varphi: D \to \R$, such that $\varphi$ is not identically zero,
we denote $D^{\min} := \min_{x \in D}\{x\}$, $D^{\max} := \max_{x \in
  D}\{x\}$, $\varphi^{\min} := \min_{x \in D} \{|\varphi(x)| :
\varphi(x) \neq 0\}$, and $\varphi^{\max} := \max_{x \in D}
\{|\varphi(x)|\}$. We write $x^+ := \max\{0, x\}$.  For any function
$\varphi : D \to \R$, $t_{\varphi}$ denotes an upper bound on the time
needed to evaluate $\varphi$ on a single point in its domain. A
function $\varphi : D \subseteq \R \to \R$ is called Lipschitz
continuous with a given Lipschitz constant $\kappa$
($\kappa$-Lipschitz continuous, in short), if $|\varphi(x) -
\varphi(y)| \le \kappa |x-y|$ for all $x, y \in D$. Given a Lipschitz
continuous function $\varphi: D \to \R$, we denote by $\kappa_\varphi$
its Lipschitz constant. Let $X$ be a set, and let $Y(x)$ be a set for
every $x \in X$. We denote by $X \otimes Y$ the set $\bigcup_{x \in X}
Y(x)$. We write $\log z$ to denote $\log_2 z$.  When indicating
running times of algorithms, we give bounds on the number of
operations, that we assume to have unitary cost. For arithmetic
operations, the number of bit operations depends on the size of the
numbers: because we always ensure that the size of the numbers is
polynomially bounded in the (binary) input size, the number of bit
operations is at most a polylogarithmic factor larger than the number
of arithmetic operations.

As mentioned in Sect.~\ref{s:techniques}, the approximation scheme
discussed in this paper relies on of $K$-approximation sets and
functions \cite{halman09}. These concepts are formalized below. The
definitions are based on \cite{halman14full,nannifptascdpfull}.
\begin{definition}
  \label{def:apxfun}
  Let $K \ge 1$ and let $\varphi : D \subset \R
  \to \R^+$. We say that $\tilde{\varphi} : D \to \R^+$ is a {\em
    $K$-approximation function} of $\varphi$ (or more briefly, a
  $K$-approximation of $\varphi$) if for all $x \in D$ we have
  $\varphi(x) \le \tilde{\varphi}(x) \le K\varphi(x)$.
\end{definition}
\begin{definition}
  \label{defn:convexext}
  Let $\varphi : D \subset \R \to \R$ be a convex function. For every
  finite $E \subseteq D$, the {\em convex extension of $\varphi$
    induced by $E$} is the function $\hat{\varphi} : [E^{\min},
    E^{\max}] \to \R$ defined as the lower envelope of the convex hull
  of $\{(x, \varphi(x)) : x \in E\}$.
\end{definition}
\begin{definition}
  \label{def:monoext}
  Let $\varphi : D \subset \R \to \R$ be nondecreasing. For every finite $E
  \subseteq D$, the {\em monotone extension of $\varphi$ induced by
    $E$} is the function $\hat{\varphi} : [E^{\min}, E^{\max}] \to
    \R$ defined as $\hat{\varphi}(x) := \min\{\varphi(y) : y \in E,
    y \ge x\}$.
\end{definition}
\begin{definition}
  \label{def:apxset}
  Let $K \ge 1$ and let $\varphi : D \subset \R \to \R^+$ be a convex
  (resp.\ nondecreasing) function. We say that a finite set $W \subset D$
  is a {\em $K$-approximation set of $\varphi$} if the convex
  extension (resp.\ the monotone extension) of $\varphi$ induced by
  $W$ is a $K$-approximation function of $\varphi$.
\end{definition}
For a monotone nondecreasing function, an algorithm to compute a
$K$-approximation set in polynomial time is described in
\cite{halman14full} if the function is defined over integers. If the
function is defined over a real interval and is strictly positive,
such an algorithm is given in \cite{nannifptascdpfull}, see the
appendix \ref{apx:functions}. We call this algorithm {\sc ApxSetInc},
consistently with these papers: whether the discrete or continuous
version should be applied will be clear from the context.
$K$-approximation functions can be combined as indicated below.
\begin{proposition} [{\bf Calculus of $K$-approximation functions}, Prop.~5.1 in \cite{halman14full}]
  \label{prp:CAF} For $i=1,2$ let $K_i \geq 1$, let $\varphi_i:D
  \rightarrow \R^+$ be an arbitrary function over domain $D$, and let
  $\tilde{{\varphi}_i}:D \rightarrow \R$ be a $K_i$-approximation of
  $\varphi_i$. Let $\psi_1:D \rightarrow D$, and let $\alpha,\beta \in
  \R^+$. Then:
  \begin{enumerate}
  \item $\varphi_1$ is a 1-approximation of itself,
  \item \label{item:CAFlin}(linearity of appr.)  $\alpha + \beta
    \tilde{\varphi_1}$ is a $K_1$-approximation of $\alpha + \beta
    \varphi_1$,
  \item \label{item:CAFsum}(summation of appr.)
    $\tilde{\varphi_1}+\tilde{\varphi_2}$ is a
    $\max\{K_1,K_2\}$-approximation of $\varphi_1 + \varphi_2$,
  \item \label{item:CAFtra} (composition of appr.)
    $\tilde{\varphi_1}(\psi_1)$ is a $K_1$-approximation of
    $\varphi_1(\psi_1)$,
  \item \label{item:CAFmin} (minimization of appr.) $\min
    \{\tilde{\varphi_1},\tilde{\varphi_2}\}$ is a
    $\max\{K_1,K_2\}$-approximation of
    $\min\{\varphi_1,\varphi_2\}$,
  \item \label{item:CAFmax} (maximization of appr.) $\max
    \{\tilde{\varphi_1},\tilde{\varphi_2}\}$ is a
    $\max\{K_1,K_2\}$-approximation of
    $\max\{\varphi_1,\varphi_2\}$,
  \item \label{item:CAFapx} (appr.\ of appr.) If
    $\varphi_2=\tilde{\varphi_1}$ then $\tilde{\varphi_2}$ is a
    $K_1K_2$-approximation of $\varphi_1$.
  \end{enumerate}
\end{proposition}

We call \emph{canonical representation} of an approximation set $W$ of
a function $\vp$ its representation as a linearly ordered set of
points and the corresponding function values. Formally, we write the
canonical representation as $\{(x,\vp(x))\; | \; x \in W\}$,
represented as an array sorted by its $x$ coordinate. Such a
representation corresponds to the data that is normally stored in a
computer implementation of the proposed algorithms. Let $\vp$ be a
convex (resp.\ monotone) function, $W$ a $K$-approximation set of
$\vp$, and let $\hat \vp$ be the convex (resp.\ monotone) extension of
$\vp$ induced by $W$. We say that $\hat \vp$ is a $K$-approximation
function of $\vp$ that is given in the form of a canonical
representation. (Note that by Def.~\ref{def:apxset}, $\hat \vp$ is
indeed a $K$-approximation function of $\vp$.) Given the canonical
representation and any query point $x$ in the domain of $\vp$, one can
compute in $O(\log |W|)$ time the value of $\hat \vp$ at
$x$. Therefore, the canonical representation serves as a value oracle
for $\hat \vp$ with query time $O(\log |W|)$. Given a canonical
representation, we can recover the approximation set $W$ on which it
was constructed in time linear in $|W|$.

We define a routine called {\sc CompressInc}($\varphi, [A,B], K$),
that given a monotone nondecreasing function $\varphi:[A,B]
\rightarrow \R^+$ and an approximation factor $K$, constructs a value
oracle (in the form of a canonical representation) for a
$K$-approximation function of $\vp$. Details are given in appendix
\ref{apx:functions}. Similar routines will be constructed in the
remainder of this paper.

The main idea for the approximation algorithm proposed in this paper
is to construct a value function approximation following the backward
recursion \eqref{eq:vf}. The value function approximation is stored as
a $K$-approximation set. At each stage, we must be able to efficiently
compute the expected value appearing in \eqref{eq:vf}, and to perform
the minimization over the action space. We will show how to achieve
these goals under the assumptions detailed in the next section.

\subsection{Assumptions}
\label{s:assumptions}
To construct an FPTAS, we require the following conditions to be
satisfied:
\begin{condition} \label{con:sets}[Domains]
  $\S_{T+1}$ and $\S_t$ for $t=1,\ldots,T$ are intervals on the real
  line. $\A_t(I_t):=\{\vx_t : A_t \vx_t \ge \vec{b}_t +
  \vec{\delta}_{b_t} I_t, \vx_t \ge 0 \} \subset \R^p$ is a
  $p$-dimensional polyhedral set that is expressed as the feasible set
  of a parametric LP with $p$ variables, where: the right-hand side
  vector is an affine function of the parameter $I_t$; $A_t
  \in \Q^{m \times p}$; and $\vec{b}_t, \vec{\delta}_{b_t} \in \Q^{m}$
  for $t=1,\ldots,T$.
\end{condition}
\begin{condition} \label{con:implicit}[Description of random events]
  For every $t=1,\dots,T+1$, we have $\vD_t \in \R^{\ell}$.  For every
  $t=1,\dots,T+1$ and $i=1,\dots,\ell$, one of the following two
  conditions hold for the $i$-th random variable $\vD_{t,i}$ of the
  vector $\vD_t$:
  \begin{itemize}
  \item[(i)] The random variable $\vD_{t,i}$ is truncated continuous
    with compact support ${\rm support}(\vD_{t,i}) = [\vD_{t,i}^{\min},
      \vD_{t,i}^{\max}] \subset \R$ of length $n_{t,i} =
    \vD_{t,i}^{\max} - \vD_{t,i}^{\min}$, and its CDF is Lipschitz
    continuous.
  \item[(ii)] The random variable $\vD_{t,i}$ is discrete, its support
     ${\rm support}(\vD_{t,i}) \subset [\vD_{t,i}^{\min}, \vD_{t,i}^{\max}] \subset
    \R$ consists of $n_{t,i} < \infty$ elements, and the $k$-th largest
    element $x$ in the support can be identified in time
    polylogarithmic in $n_{t,i}$.
  \end{itemize}
  Furthermore, $\Pr(\vD_{t,i} = \vD_{t,i}^{\min}) > 0$, $\Pr(\vD_t
  = \vD_{t,i}^{\max}) > 0$, and the information about the random
  events is given via value oracles to the CDF. All $\vD_{t,i}$ are
  independent.
\end{condition}
\begin{condition} \label{con:functions}[Structure of the functions]
  For every $t=1,\ldots,T$, the function $f_t(I_t,\vx_t,\vD_t) : \S_t
  \otimes \A_t \times \D_t \to \S_{t+1}$ is linear in its variables
  and is expressed as $f_t(I_t, \vx_t, \vD_t) := \theta^I I_t +
  \vec{\theta}^x \cdot \vx_t + \vec{\theta}^D \cdot \vD_t$. The
  function $g_t$ can be expressed as
  $g_t(I_t,\vx_t,\vD_t)=g_t^I(I_t,\vx_t)+g_t^D(f^g_t(I_t,\vx_t,\vD_t))$,
  where: $g_t^I : \S_t \otimes \A_t \to \R^+$ is a piecewise linear
  convex function described as the pointwise maximum of a set of $q_t
  < \infty$ given hyperplanes; and $f^g_t : \S_t \otimes \A_t \times
  \D_t \to \G_t$ is an affine function expressed as $f^g_t(I_t, \vx_t,
  \vD_t) := \sigma^I I_t + \vec{\sigma}^x \cdot \vx_t + \vec{\sigma}^D
  \cdot \vD_t$. One of the following conditions hold:
  \begin{itemize}
    \item[(i)] The function $g_{T+1} : \S_{T+1} \to \R^+$ is positive
      piecewise linear convex over $\S_{T+1}$ with $m_{T+1}$ given
      breakpoints and slopes, and minimum value $g_{T+1}^{\min} >
      0$. Furthermore, for all $t=1,\ldots,T$ the function
      $g_t^D(\cdot) : \G_t \subset \R \to \R^+$ is a piecewise linear
      convex function over a compact interval $\G_t$ with $m_t$ given
      breakpoints and slopes.
    \item[(ii)] The function $g_{T+1} : \S_{T+1} \to \R^+$ is
      Lipschitz continuous and convex over $\S_{T+1}$ and is described
      via a value oracle. There is a given positive bound $\tilde
      g_{T+1}^{\min} > 0$ on the minimum value of
      $g_{T+1}(\cdot)$. Furthermore, for all $t=1,\ldots,T$ the
      function $g_t^D(\cdot) : \G_t \subset \R \to \R^+$ is
      non-negative Lipschitz continuous and convex over a compact
      interval $\G_t$, and is described via a value oracle.
  \end{itemize}
\end{condition}
We give next a few remarks about our DP model.  In Condition
\ref{con:implicit}(i) and Condition~\ref{con:functions}(ii) the
Lipschitz constant need not be given a-priori. However, the running
time of our algorithms depends polylogarithmically on it. In Condition
\ref{con:implicit}(ii), we use $n_{t,i}$ to indicate either the size
of the support in terms of interval length for continuous random
variables, or the cardinality of the support for discrete random
variables: this allows us to write the running time in a more compact
way. An example of a truncated continuous r.v.\ satisfying
Condition~\ref{con:implicit} is that of a Gaussian r.v.~$Y$ with given
mean and standard deviation, clipped to a bounded interval: for
example, if $Y$ represents a demand, we can truncate it from below at
zero because demand cannot be negative ($\Pr(X=0)=\Pr(Y <0)$), and
from above at some positive order capacity $u$ because all demand
above this value is treated identically
($\Pr(X=u)=\Pr(Y>u)$). Truncated mixed r.v.s can in principle be
handled under assumptions similar to Condition \ref{con:implicit}, but
their treatment is cumbersome and is not pursued here. We remark that
the implicit model for the description of the r.v.s, as stated in
Condition \ref{con:implicit}, is both more general and more compact
than listing all possible values in the support of the r.v.s and the
corresponding probabilities. The latter approach, while simpler and
more commonly employed, would only be viable for discrete r.v.s with
support of small size.

Condition \ref{con:functions}(i) is the traditional model in which the
cost functions are analytically known in explicit form. Condition
\ref{con:functions}(ii) is a more flexible model that allows some of
the cost functions to be described by value oracles. The FPTAS
proposed in this paper is essentially the same under both Condition
\ref{con:functions}(i) and \ref{con:functions}(ii), but some of our
hardness results assume the oracle model. Studying the oracle model is
appealing because positive results that hold in this setting have weak
assumptions and can therefore be broadly applied. Value oracles can
also allow strong negative results, see Sect.~\ref{s:hardness}.


The input data of the problem includes: the number of time periods $T$;
the initial state $I_1$; the constant $\gamma=\min
\{\Pr(\vD^{\max}_{t,i}),\Pr(\vD^{\min}_{t,i}) \;|\; t=1,\ldots,T,
i=1,\ldots,\ell\}$ (Condition \ref{con:implicit} implies $\gamma>0$);
the maximum cost in a time period $U_g$; the maximum Lipschitz
constant $\kappa$ of the cost functions $g_t^I, g_t^D$ (and of the CDF
of $\vD_{t,i}$ in the case of Condition \ref{con:implicit}(i)) for
$t=1,\dots,T+1$; the minimum value $\mu := g_{T+1}^{\min}$ of the
terminal cost function; the maximum length $U_S$ of the state spaces;
the maximum size $U_A$ of the coefficients in the constraint matrix or
right-hand side vector of the LPs describing $A_t(I_t)$; the maximum
size $U_f$ of the coefficients in the vectors $\vec{\theta}^x,
\vec{\sigma}^x, \vec{\theta}^D, \vec{\sigma}^D$ and of the scalars
$\theta^I,\sigma^I$ in the description of $f_t$ and $f^g_t$; the
maximum numbers of pieces $m^* = \max_t m_t$ and $q^* = \max_t q_t$ in
the piecewise description of the cost functions; the maximum size of
the support $n^* = \max_{t,i} n_{t,i}$ of the r.v.s. Our algorithms
run in time polynomial in $\log I_1$, $\log 1/\gamma$, $\log 1/\mu$,
$\log U_g$, $\log \kappa$, $\log U_S$, $\log U_A$, $\log U_f$, $m^*$,
$q^*$, $p$, $\ell$, $T\log n^*$. Throughout this paper, the Lipschitz
constants of the functions involved need not be known, but the running
time of the algorithms depends on them: intuitively, a function with
high Lipschitz constant may change quickly, and it may take longer for
the algorithms to identify sudden ``jumps'' in function value with the
required precision. The necessity of our assumptions is discussed in
Sect.~\ref{s:hardness}.

\subsection{An application to resource management under uncertainty}
\label{s:application}
We describe here an example of the type of models allowed under the
assumptions given in the previous section. The example involves managing
a resource of non-discrete nature, e.g., liquid natural gas or cement,
at a central storing facility. Let $G = (V, E)$ be a directed graph
with $|V| = m + 1$ nodes with a special node $v_0$ that is the central
storing facility; for $v \in V$, we denote by $\delta^+(v)$ its set of
outgoing edges and by $\delta^-(v)$ its set of ingoing edges. The
optimization is performed over a finite time horizon $T$. The state of
the system $I_t$ is the amount of resource available at the central
storing facility at the beginning of time period $t$.  At the end of
every time period, there is an unknown arrival of the resource encoded
by a r.v.\ $\vD_{t,0}$, e.g., cargo ships (in which case, the r.v.\ is
discrete). There are $m$ locations, $v_1,\dots,v_m$ that consume the
resource but do not have (between time periods) storing capabilities.
Each location has unknown demand at time period $t$, distributed according
to a continuous r.v.\ $\vD_{t,i}$ with $i=1,\dots,m$. The demand
cannot be transferred to other locations: unsatisfied demand or excess
inventory at each of the locations is lost.  The resource can be moved
from the central storing facility to the $m$ locations, as well as
between locations, over a capacitated flow network described by $G$,
potentially with losses $1-a_{ij}$ over the arcs $(i,j) \in E$ (i.e.,
the flow variables may appear with coefficients $< 1$ in the flow
balance constraints at some nodes). There is no backlogging at the
central storing facility: inventory must stay nonnegative (but we
could easily amend the formulation to allow backlogging). At each time
period $t$, one observes the state of the system $I_t$, and based on
the distribution of the demand in time periods $t,t+1,\ldots,T$
decides on the flow $f_{t,ij}$ over each arc $(i,j) \in E$. At the end
of the period a cost is incurred that consists of the transportation
cost given by the flow network with unit cost $c_{ij}$ on each arc,
and shortage (resp.\ disposal) cost $s_i$ (resp.\ $h_i$) that is
accounted for every unit by which the demand $\vD_{t,i}$ at location
$i$ is not met (resp.\ is exceeded). At the final time period $T+1$,
there is a disposal cost $h_0$ for every unit left in the
inventory. We denote by $y_{t,i}$, $i=0,\dots,m$ the flow balance
(exiting minus entering) at location $i$, therefore $y_{t,i}$ is
positive if the node sends out more than it receives, negative
otherwise. Furthermore, we denote by $w_{t,i}$ an auxiliary variable
to encode the cost incurred at the end of time period $t-1$ at
location $i$, i.e., the maximum between the shortage and disposal
costs.

This problem can be formulated as a multistage stochastic linear
program as follows:
\begin{alignat}{3}
  \min \; && \displaystyle \mathbb{E}_{\vD} \left[\sum_{t=1}^{T} \left(\sum_{(i,j) \in E} c_{ij} f_{t,ij}(\vD_t) + \sum_{i=1}^m w_{t+1,i}(\vD_t) \right) + h_0 I_{T+1}(\vD_{T}) \right] \label{eq:lpobj} \\
  t=1,\dots,T && - I_t(\vD_{t-1}) + I_{t+1}(\vD_t) + y_{t,0}(\vD_t) = \vD_{t,0} \label{eq:lpinventory} \\
  t=1,\dots,T, i=0,\dots,m && \sum_{(i,j) \in \delta^+(i)} f_{t,ij}(\vD_t) - \sum_{(j,i) \in \delta^-(i)} a_{ji} f_{t,ji}(\vD_t) = y_{t,i}(\vD_t) \label{eq:lpbalance} \\
  t=1,\dots,T && y_{t,0}(\vD_t) \leq I_t(\vD_{t-1}) \label{eq:lpnonnegativeinventory} \\
  t=2,\dots,T+1, i=1,\dots,m && -h_i y_{t-1,i}(\vD_{t-1}) - h_i \vD_{t-1,i} \le w_{t,i}(\vD_t) \label{eq:lpholding} \\
  t=2,\dots,T+1, i=1,\dots,m && s_i y_{t-1,i}(\vD_{t-1}) + s_i \vD_{t-1,i} \le w_{t,i}(\vD_t) \label{eq:lpshortage} \\
  t=1,\dots,T, (i,j) \in E && f_{t,ij}(\vD_t) \ge 0 \label{eq:lpnonneg1} \\
  t=1,\dots,T, i=1,\dots,m && y_{t,i}(\vD_t) \le 0 \label{eq:lpnonneg2} \\
  && I_1(\vD_0) = \bar{I}_1. \label{eq:lpinitcond}
\end{alignat}
In the above formulation, $\bar{I}_1$ is the initial inventory. The
inventory level $I_{t+1}(\vD_{t})$ depends on the random variable at
time period $t$, as indicated in constraint \eqref{eq:lpinventory}:
$I_{t+1}$ is automatically determined after all decisions at time
period $t$ are taken and $\vD_{t}$ is revealed. Similarly, the
auxiliary variable $w_{t,i}(\vD_t)$ encodes the cost at location $i$
at the end of the previous time period, because such cost can only be
computed after the demand at a given time period is revealed. $\vD_0$
only serves the purpose of keeping this notation consistent for
$I_1(\vD_0)$, and is otherwise not used in the model. Constraint
\eqref{eq:lpnonnegativeinventory} implies $I_t(\vD_{t-1}) \ge 0$.

If the $\vD_{t,i}$ are all discrete r.v.s, this problem admits a
deterministic equivalent formulation by introducing a copy of the
decision variables for each sample path, but the number of sample
paths is exponential in $T$. We now show that a DP formulation of this
problem satisfies Conditions \ref{con:sets}-\ref{con:functions}.

Let $I_t$ be the state variable. The action vector is defined by
$\vx_t = \left((y_{t,i})_{i=1,\dots,m}, (f_{t,ij})_{(i,j) \in
  E}\right)$; for simplicity, we use $y_{t,i}, f_{t,ij}$ to refer to
the variables in the action vector $\vx_t$ corresponding to the
variables with the same name in
\eqref{eq:lpobj}-\eqref{eq:lpinitcond}.  The transition function is
$f_t(I_t, \vx_t, \vD_t) = I_t - y_{t,0} + \vD_{t,0}$. The action space
is $\A_t(I_t) = \{ \vx_t : \sum_{(i,j) \in \delta^+(i)} f_{t,ij} -
\sum_{(j,i) \in \delta^-(i)} a_{ji} f_{t,ji} = y_{t,i} \; \forall
i=0,\dots,m; y_{t,0} \le I_t; f_{t,ij} \geq 0 \; \forall (i,j) \in E;
y_{t,i} \le 0 \; \forall i=1,\dots,m\}$. The cost functions are:
$g^I_t(I_t, \vx) = \sum_{(i,j) \in E} c_{ij} f_{t,ij}$, $g^D_t(I_t,
\vx, \vD_t) = h_i (-y_{t,i} - \vD_{t,i})^+ + s_i (y_{t,i} +
\vD_{t,i})^+$, $g_{T+1}(I_{T+1}) = h_0 I_{T+1}$. It is easy to verify
that with the cost functions defined as given, we obtain an equivalent
description to problem \eqref{eq:lpobj}-\eqref{eq:lpinitcond}, and
Conditions \ref{con:sets}-\ref{con:functions}(i) are satisfied.


\section{Hardness results}
\label{s:hardness}
We now discuss the necessity of our assumptions. Regarding Condition
\ref{con:sets}, it is an open question under what additional assumptions the restriction that
the state spaces $\S_t$ are intervals on the real line can be lifted
under our DP model. However, the next subsections show
that constructing an approximation of the value function
is difficult already whenever $\S_t$ are two-dimensional boxes (Cor.~\ref{cor:2ddphard}).
We show this by first proving the existence of a class of two-dimensional piecewise linear convex functions that are hard to approximate.

Regarding Condition \ref{con:implicit}, the conditions on truncated
r.v.s seem difficult to relax under the oracle model for the CDFs:
\cite[Prop.~2.6]{nannifptascdpfull} shows that a continuous function
over a compact interval may not admit an efficient approximation
unless its values at the endpoints of the domain are bounded away from
zero. This result holds even if the function is monotone and Lipschitz
continuous. Hence, approximating the CDF of a continuous r.v.\ that
does not have positive probability mass at the endpoints seems
difficult. (\cite{nannifptascdpfull} also shows that for bounded
Lipschitz continuous functions we can drop the strict nonnegativity
requirement if we allow both relative and absolute error at the same
time; however, in this paper we aim for an FPTAS in the usual sense.)
Regarding Condition \ref{con:functions}, the nonnegativity requirement
on the cost functions $g_t$ cannot be relaxed:
\cite{nannienergystorage} exhibits a DP that is \#P-hard to
approximate to any constant factor, and their problem can be cast into
our framework if we allow the cost functions to be of either sign. In
other words, the \#P-hard DP in \cite{nannienergystorage} satisfies
Conditions \ref{con:sets}-\ref{con:functions} except for the
nonnegativity requirement on $g_t$. It is an open question whether the
restriction $g_{T+1}^{\min} > 0$ (or $\tilde g_{T+1}^{\min} > 0$ in
the setting of Condition~\ref{con:functions}(ii)) can be relaxed to
$g_{T+1}^{\min} \ge 0$ (to $\tilde g_{T+1}^{\min} \ge 0$ in the
setting of Condition~\ref{con:functions}(ii)). The difficulty stemming
from $g_{T+1}^{\min} = 0$ is that we would have to keep track of the
exact location of the zeroes of the value function at each stage.

\subsection{Hardness of approximation of multivariate convex functions}
\label{sec:multihard}
The DP model discussed in this paper naturally yields piecewise linear
convex value functions. To extend the framework to a multidimensional
state space setting we would have to efficiently build approximations
of such functions. However, the next results show that some piecewise
linear convex functions do not admit an efficient approximation. We
first show that there is an exponential number of two-dimensional
piecewise linear convex functions with distinct values at integer
points, and use this to prove that approximations of such functions
require space exponential in the size of the domain.

\begin{theorem}
  \label{thm:bivarhard}
  For any nonnegative integer $A$, there exist $\Omega(2^{\sqrt{U}})$
  distinct nondecreasing piecewise linear convex functions $\varphi :
  [1, U]^2 \to [A, A+U]$ that attain the value $A$ at distinct sets of
  points in $[1,\dots,U]^2$, and at least $A + 1$ at the other sets of
  integer points. These functions are described as the pointwise
  maximum of $O(\sqrt{U})$ hyperplanes,
\end{theorem}
\begin{proof}
  Let
  \begin{equation*}
    \Phi := \left\{\varphi_{r_1,\dots,r_U} : [1, U]^2 \to [A, A+U] \middle| r_1,\dots,r_U \in [0, U]; \sum_{i=1}^U r_i = U, r_1 \le r_2 \le \dots \le r_U \right\}
  \end{equation*}
  be a family of bivariate discrete functions defined as follows: for $x,y \in
  [1,\dots,U]$,
  \begin{equation}
    \label{eq:varphidef}
    \varphi_{r_1,\dots,r_U}(x,y) =
    \begin{cases}
      A & \text{if } x = 1,\dots,U - \sum_{i=1}^y r_i \\
      A + x - U + \sum_{i=1}^y r_i & \text{otherwise}.
    \end{cases}
  \end{equation}
  \cite[Prop.~3.4]{H15} shows (taking $A = 0$) that any pair of
  elements of $\Phi$ attains the value $A$ at different points on the
  integer grid. Furthermore, it shows that the number of such
  functions is equal to the number of partitions of the integer $U$,
  which is bounded below by $\frac{H}{U}e^{2\sqrt{U}} \ge H
  2^{\sqrt{U}}$. In Fig.~\ref{fig:varphiex} we illustrate the shape of
  $\varphi_{r_1,\dots,r_U}$ by plotting its values for two choices of
  $r_1,\dots,r_U$.

  \begin{figure}[tb]
    \centering
    \includegraphics[width=0.4\textwidth]{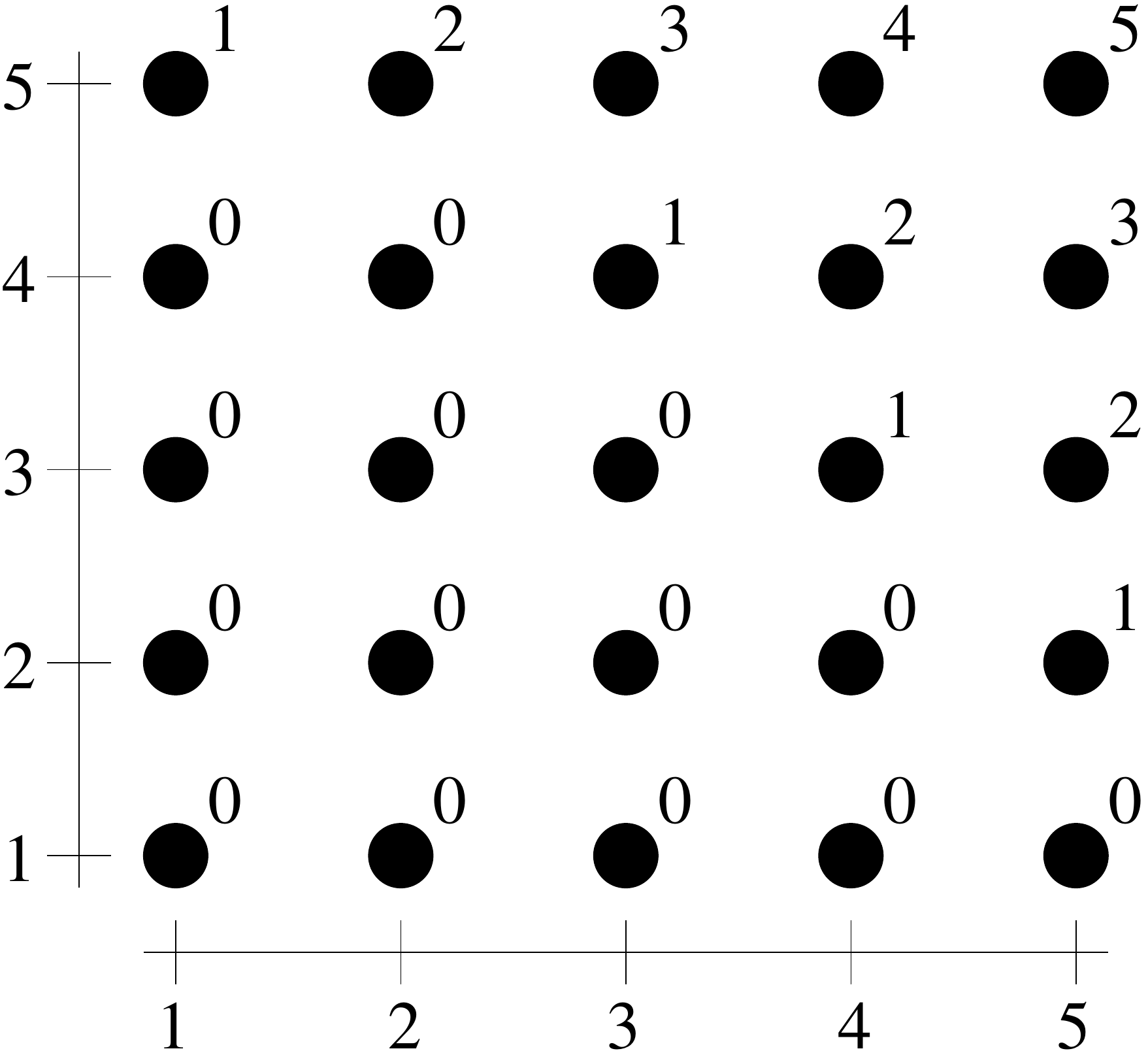}\hspace{5em}
    \includegraphics[width=0.4\textwidth]{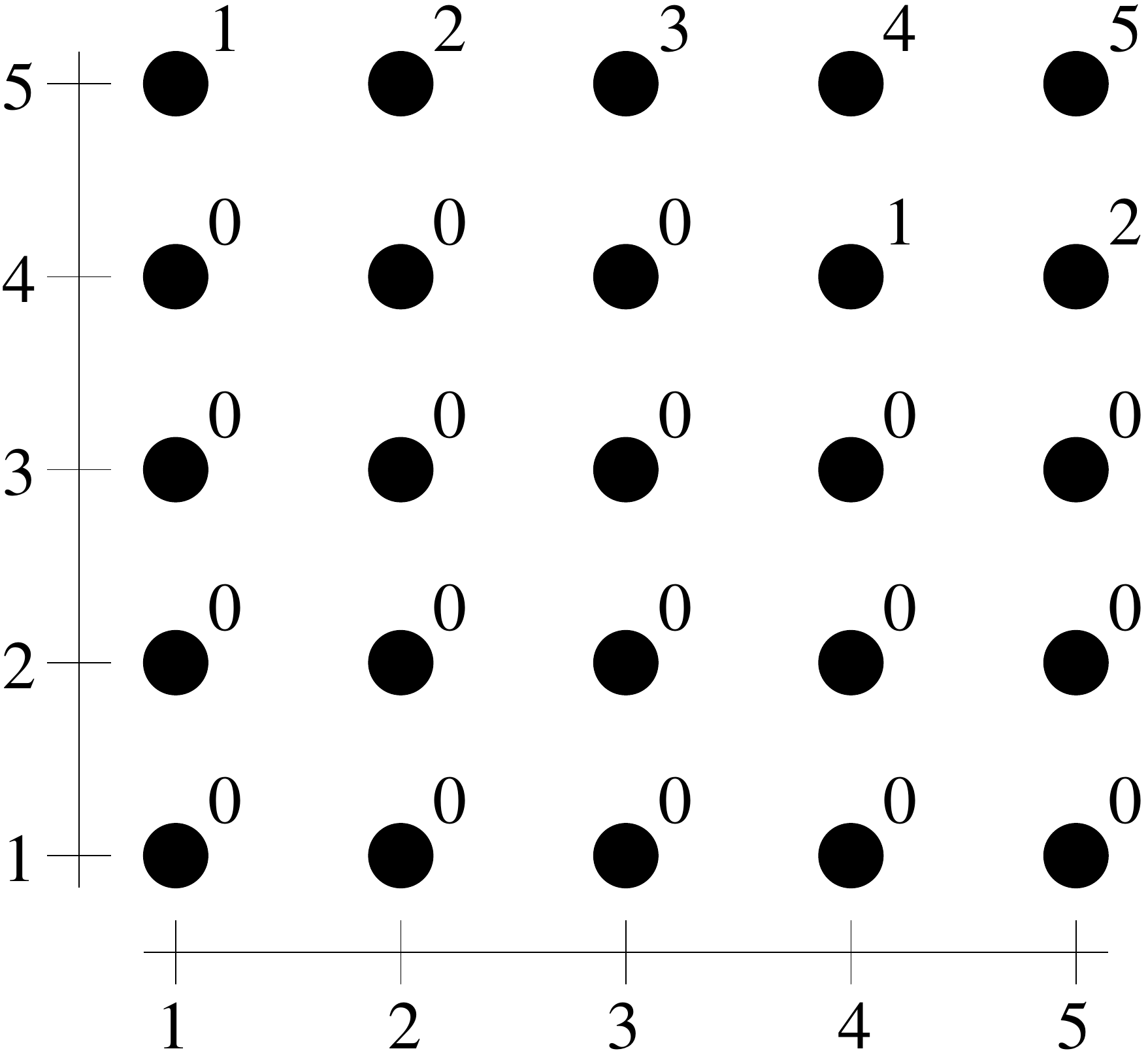}
    \caption{Values of the function $\varphi_{0,1,1,1,2}$ (left) and
      $\varphi_{0,0,0,2,3}$ (right) on the integer grid
      $[1,\dots,5]^2$, with $A = 0$. At each integer point in the
      plane, we indicate the corresponding function value next to it.}
    \label{fig:varphiex}
  \end{figure}

  Let $\text{diff}(r_1,\dots,r_U)$ be the number of distinct values
  among $r_1,\dots,r_U$. We show next that each function
  $\varphi_{r_1,\dots,r_U}$ can be described as the pointwise maximum
  of at most $\text{diff}(r_1,\dots,r_U) + 1$ hyperplanes, in a such a
  way that the value of the pointwise maximum on the integer grid
  matches that of $\varphi_{r_1,\dots,r_U}$.
  Note that the slopes of $\varphi_{r_1,\dots,r_U}$ satisfy the following relationship:
  \begin{align}
    \label{eq:xslope}
    \varphi_{r_1,\dots,r_U}(x+1, y) - \varphi_{r_1,\dots,r_U}(x, y) &=
    \begin{cases} 0 & \text{if } x < U - \sum_{i=1}^y r_i \\
      1 & \text{otherwise,}
    \end{cases} \\
    \label{eq:yslope}
    \varphi_{r_1,\dots,r_U}(x, y+1) - \varphi_{r_1,\dots,r_U}(x, y) &=
    \begin{cases} 0 & \text{if } x \le U - \sum_{i=1}^{y+1} r_i \\
      r_{y+1} + \min\{x - U + \sum_{i=1}^y r_i, 0\} & \text{otherwise.}
    \end{cases}
  \end{align}
  Furthermore, $\varphi_{r_1,\dots,r_U}(x,U) = A + x$ for all valid
  choices of $r_1,\dots,r_U$.

  {\em Claim.} For $1 \le k \le U - 1$, there exists $S_k := \{(a, b,
  c) \in \N\}$, with $|S_k| \le \text{diff}(r_k,\dots,r_U) + 1$, such that
  $\varphi_{r_1,\dots,r_U}(x, y) = \max_{(a,b,c) \in S_k}\{ax +by +
  c\}$ for all $1 \le x \le U, y \ge k$.

  {\em Proof of the claim.} By backward induction. We start with $k =
  U-1$. Consider the set $S_k := \{(0, 0, A), (1, r_U, -r_UU)\}$,
  which corresponds to two hyperplanes given in the form $ax + by +
  c$: one ``flat'' hyperplane at the value $A$, and one hyperplane
  with slope $1$ along the $x$ axis, $r_U$ along the $y$
  axis. Consider the latter hyperplane, which we call $H$. By the
  choice of $c = -r_UU$, $H$ interpolates $\varphi_{r_1,\dots,r_U}$ at
  the points $(x,U)$ for all $1 \le x \le U$. Furthermore, by
  \eqref{eq:yslope}, it interpolates $\varphi_{r_1,\dots,r_U}$ at the
  points $(x, U-1)$ for which $x \ge U - \sum_{i=1}^{k} r_i$, because
  for such points its slope on the $y$ direction matches the slope of
  $\varphi_{r_1,\dots,r_U}$. It remains to analyze the points $(x,
  U-1)$ with $x < U - \sum_{i=1}^{k} r_i$. Let $P$ be the set of such
  points. For $(x,y) \in P$, $\varphi_{r_1,\dots,r_U}(x,y)$ has value $A$
  by definition, see \eqref{eq:varphidef}. Furthermore,
  $\varphi_{r_1,\dots,r_U}(x,y+1) = \varphi_{r_1,\dots,r_U}(x,U) = x +
  A < U - \sum_{i=1}^{k} r_i + A = r_U + A$. Then, since $H$
  interpolates $\varphi_{r_1,\dots,r_U}$ at $(x,y+1)$, and its slope
  is $r_U$ in the $y$ direction, we have that the value of $H$ at
  $(x,y)$ is equal to $x + A - r_U < A$. But then the maximum of the
  two hyperplanes in $S_k$ at $(x,y) \in P$ is given by $A$ (i.e., it
  is given by the ``flat'' hyperplane), which is exactly the value of
  $\varphi_{r_1,\dots,r_U}(x,y)$ by equation
  \eqref{eq:varphidef}. This concludes the initial step of the
  induction.

  We now need to prove the induction step $k$, assuming the claim is
  true for $k+1$. Suppose first that $r_{k} = r_{k+1}$. In this
  case, a similar argument to the one above shows that no further
  hyperplane is needed in $S_k$, because the slope of
  $\varphi_{r_1,\dots,r_U}$ along the $y$ axis for points with $x \ge
  U - \sum_{i=1}^{k} r_i$ and $y=k$ is the same as for $y=k+1$, and
  for the remaining points, the ``flat'' hyperplane suffices. We need
  to analyze the case $r_{k} < r_{k+1}$. In this case, let $S_k =
  S_{k+1} \cup \{(1, r_{k+1}, \varphi_{r_1,\dots,r_U}(U,k+1) - U -
  r_{k+1}(k+1))\}$. Let $H$ be the hyperplane defined by the
  coefficients $(1, r_{k+1}, \varphi_{r_1,\dots,r_U}(U,k+1) - U -
  r_{k+1}(k+1))$; it is easy to verify that $H$ interpolates
  $\varphi_{r_1,\dots,r_U}$ at $(x,k+1)$ for all $1 \le x \le U$ by
  construction. The argument now proceeds as before: the hyperplane
  interpolates $\varphi_{r_1,\dots,r_U}$ at points $(x, k)$ for which
  $x \ge U - \sum_{i=1}^{k} r_i$, because for such points its slope on
  the $y$ direction matches the slope $r_{k+1}$ of
  $\varphi_{r_1,\dots,r_U}$. The remaining points with $y=k$ have
  value $A$, therefore the ``flat'' hyperplane interpolates them and
  $H$ has value $< A$ at them. Furthermore, because the $r_{k+1} <
  r_{k+2} \le \dots \le r_U$, the value of $H$ at all points $(x,y)$
  with $y > k$ is smaller than $\varphi_{r_1,\dots,r_U}(x,y)$. This
  concludes the inductive claim.

  The above claim shows that we can construct a set $S_1$ containing
  hyperplanes such that for every $(x,y) \in [1,\dots,U]^2$, the
  pointwise maximum of the hyperplanes at $(x,y)$ is equal to
  $\varphi_{r_1,\dots,r_U}(x,y)$. Note that all hyperplanes except the
  ``flat'' hyperplane have positive integer slopes. Thus, the
  pointwise maximum of the hyperplanes at integer points is
  integer-valued. This implies that the value of the pointwise maximum
  of such hyperplanes on the integer grid is either $A$ or $\ge A+1$.
  To conclude the proof, we need a bound on $|S_1| =
  \text{diff}(r_1,\dots,r_U) + 1$. Recall that $\sum_{i=1}^U r_i = U,
  r_1 \le r_2 \le \dots \le r_U$. The largest possible number of
  distinct values of $r_i$ is attained when all $r_i$ that increase,
  increase by exactly one. This implies that it is the largest number
  $d$ such that $\sum_{j=1}^d j = \frac{d}{2}(d+1)\le U$. Then,
  $\text{diff}(r_1,\dots,r_U) = O(\sqrt{U})$.
\end{proof}

\begin{corollary}
  \label{cor:2dconvexhard}
  Let $A$ and $U$ be nonnegative integers, and let $\Phi$ be a family
  of nondecreasing convex functions $\varphi: [1, U]^2 \to [A, A+U]$
  that are described as a pointwise maximum of hyperplanes. Then, any
  approximation of $\varphi \in \Phi$ that attains relative error less
  than $\frac{A+1}{A}$, or absolute error less than $1$, requires
  $\Omega(\sqrt{U})$ space regardless of the scheme used to represent
  the function.
\end{corollary}
\begin{proof}
  By Thm.~\ref{thm:bivarhard}, there are $\Omega(2^{\sqrt{U}})$
  different nondecreasing convex functions that attain the value $A$
  at distinct points in $[1,\dots,U]^2$, and value at least $A+1$ at
  the remaining integer points. Any approximation of $\varphi$ that
  attains relative error less than $\frac{A+1}{A}$, or absolute error
  less than $1$, characterizes exactly the points at which $\varphi$
  attains the value $A$. Thus, the approximation must distinguish
  between all functions in the family $\Phi$ described in the proof of
  Thm.~\ref{thm:bivarhard}. This implies that every function in $\Phi$
  must have a different description, regardless of the scheme used to
  describe a function. Since $|\Phi| \in \Omega(2^{\sqrt{U}})$, it
  follows that any approximation of $\varphi$ must take
  $\Omega(\sqrt{U})$ space.
\end{proof}
As mentioned in Sect.~\ref{s:literature}, there exist several
special classes of multivariate discrete convex functions; the
inclusion relationships among these classes are depicted in
Figure~\ref{fig:murota}. Corollary \ref{cor:2dconvexhard} has
implications on the approximability of convex-extensible functions,
which we define next. Let $f:\Z^d \ra \R^+$ be a function defined over
the $d$-dimensional lattice of integer points. The convex closure of
$f$ is defined to be a function $\bar f : \R^n \ra \R^+$ given by
\begin{equation*}
\bar f(x) = \sup_{p \in \R^n, \alpha \in \R} \{ \langle p,x \rangle +\alpha \; | \;  \langle p,y \rangle +\alpha \leq f(y), \; \forall y \in \Z^n\} \quad \forall x \in \R^n.
\end{equation*}
If this function $\bar f$ coincides with $f$ on integer points, i.e.,
if $\bar f(x)=f(x) \quad \forall x \in \Z^n$, we say that $f$ is {\em
  convex extensible} and call $\bar f$ the {\em convex extension of}
$f$ \cite[Chap. 3.4]{Mu03}.

The starting point for the proof of Theorem~\ref{thm:bivarhard} is the
family $\Phi$ of bivariate discrete convex functions (in the sense of
Miller \cite{Mi71}) over domain $[1,\ldots,U]^2$ as defined in the
proof of \cite[Prop. 3.4]{H15}. For each function $f \in \Phi$ we
construct a set $S_1$ of $O(\sqrt{U})$ hyperplanes and define a
(continuous) convex function $\bar f:[1,U] \ra \R^+$ as the pointwise
maximum of the hyperplanes in $S_1$. We then show that $f$ coincides
with $\bar f$ on the integer square $[1,\ldots,U]^2$ and therefore the
family $\Phi$ consists of discrete functions that are
convex-extensible (and are also convex in the sense of Miller at the
same time). Using the fact that $|\Phi|=\Omega(\sqrt{U})$
\cite[Prop. 3.4]{H15}, and following similar arguments as in the proof
of Corollary~\ref{cor:2dconvexhard}, Theorem~\ref{thm:bivarexthard}
follows. We remark that the functions in \cite[Prop. 3.4]{H15} are
discrete convex as opposed to convex in the ordinary sense, and the
two concepts do not necessarily coincide; this paper settles the
status of discrete convex-extensible and continuous (ordinary)
piecewise linear convex functions, showing their inapproximability in
space polynomial in the (binary) size of the domain.


\subsection{Hardness of multidimensional DP}
\label{s:oraclehardness}


In this section we make some remarks about the approximability of
\eqref{eq:dp} when functions $g_{T+1}, g_t^D$ are described by a value
oracle and are not necessarily piecewise linear, as in
Condition~\ref{con:functions}(ii). We focus on these two components of
the cost function because the remaining terms $g^I_t$ appear in the
minimization over an $\ell$-dimensional polyhedron; achieving an FPTAS
for the case with $g^I_t$ described by a value oracle would therefore
require many additional mathematical tools, and is outside the scope
of this paper.

As a consequence of Cor.~\ref{cor:2dconvexhard}, a two-dimensional
generalization of problem \eqref{eq:dp} is hard to approximate under
the oracle model. We formalize this below. Consider the problem:
\begin{align}
  z^\ast(\vec{I}_1) = \min_{\pi_1,\dots,\pi_T}
  \mathbb{E} \left[\sum_{t = 1}^T g_t(\vec{I}_t, \pi_t(\vec{I}_t), \vD_t) +
    g_{T+1}(\vec{I}_{T+1}) \right], \tag{MD-DP} \label{eq:vecdp}\\
  \text{subject to: } \vec{I}_{t+1} = \vec{f}_t(\vec{I}_t, \pi_t(\vec{I}_t), \vD_t), \qquad
  \forall t=1,\dots,T, \notag
\end{align}
where $\vec{I}_t \in \R^d$, $\vec{f}_t$ is a vector-valued function,
each component of $\vec{f}_t$ is affine, $g_{T+1}$ is strictly
positive convex, and $g_t$ can be expressed as
$g_t(\vI_t,\vx_t,\vD_t)=g_t^I(\vI_t,\vx_t)+g_t^D(
\vec{f}^g_t(\vI_{t,i},\vx_t,\vD_t))$, where $g_t^I$ and $g_t^D(\cdot)$ are
nonnegative convex, and $\vec{f}^g_{t}$ has $d$ components each of
which is affine. Then, if $g_{T+1}$ or $g_t^D$ are described by value
oracles, Cor.~\ref{cor:2dconvexhard} immediately implies the
following.

\begin{corollary}
  \label{cor:2ddphard}
  There does not exist a PTAS for continuous dynamic programs of the
  form \eqref{eq:vecdp} with cost functions described by a value
  oracle, even if $d=2, T=0$, $g_1$ is nonincreasing, piecewise
  linear convex and bounded from below by a given positive number.
\end{corollary}
By contrast, when $d=1$ this paper provides an FPTAS under the same
assumptions.

\subsection{Hardness of approximating the value function of two-stage stochastic linear programs}

In addition to DP, a widely used paradigm for optimization under
uncertainty is stochastic programming. Due its importance and
applicability, we find it worthy to report that
Thm.~\ref{thm:bivarhard} has implications on the hardness of
approximation of stochastic programs. The concept of value functions
is widely used in stochastic programming just as in DP. It is well
known that the value function of a stochastic linear program may have
a number of hyperplanes that is proportional to the number of values
in the support of the r.v.s. In this section we show that if there are
two variables linking consecutive stages, even an approximation of the
value function may require a large number of hyperplanes. In
particular, if the r.v.s are described in a compact form, i.e., as a
value oracle to their CDF, then such number is exponential in the
input size: this is formalized in the next theorem.  By contrast, if
there is only one linking variable, we can construct such an
approximation, with the approach discussed in the previous section.

\begin{theorem}
  \label{thm:2sslphard}
  There does not exist a PTAS to compute the optimal value function
  for the second stage of a two-stage stochastic linear program with a
  single discrete r.v.\ described via an oracle to its CDF, even if
  there are only two variables linking the first and second stage,
  and the value function is bounded from below by a given positive
  number.
\end{theorem}
\begin{proof}
  We show how to construct a two-stage stochastic linear program such
  that, for any partitioning $r_1,\dots,r_U$ of a given integer $U$,
  the optimal second-stage value function $Q(x,y)$ is exactly the
  function $\varphi_{r_1,\dots,r_U}$ described in the proof of
  Thm.~\ref{thm:bivarhard}. At the same time, the input size is
  $O(\log U)$. By Cor.~\ref{cor:2dconvexhard}, this implies that there
  cannot be a PTAS for the problem.

  Let the first stage of the problem be the following problem:
  \begin{equation*}
    \begin{array}{rrcl}
      \min & c_1 x + c_2 y + Q(x,y) && \\
      & B \begin{pmatrix} x \\ y \end{pmatrix} &=& \vec b \\
      & x, y &\in& [1,U],
    \end{array}
  \end{equation*}
  where $Q(x,y)$ is the second-stage value function, $B$ is a suitable
  constraint matrix and $\vec b$ is a corresponding
  r.h.s.\ vector. Define the second stage problem as:
  \begin{equation}
  \begin{array}{rrcl}
    Q(x,y) := \min & z(D) + A & & \\
    \text{s.t.:} & w(D) &\ge& y - D \\
    & z(D) &\ge& U\mathbb{E}_D[w(D)] + x - U \\
    & w(D) &\ge& 0 \\
    & z(D) &\ge& 0.
  \end{array}
  \label{eq:2sslp}
  \end{equation}
  This is a stochastic linear program with an expected-value
  constraint. In \eqref{eq:2sslp}, $A > 0$ is the minimum of the value
  function.

  For a given partition $r_1,\dots,r_U$ of the integer $U$, the slopes
  of $\varphi_{r_1,\dots,r_U}$ along its two dimensions are given in
  \eqref{eq:xslope} and \eqref{eq:yslope}. Notice that the $x$-axis
  slope is always 0 or 1. Let $b_1=1,\dots,b_m \in [1,\dots,U-1]$ be the
  integer values at which the slope of $\varphi_{r_1,\dots,r_U}(U,y)$
  changes along the $y$ axis. We define $b_{m+1} = U$ for notational
  convenience. For $i=1,\dots,m$, let $\Delta_i$ be the $y$-slope of
  $\varphi_{r_1,\dots,r_U}(U,y)$ in $[b_i, b_{i+1}]$, which is
  constant by definition of $b_i$. We define $\Delta_0 = 0$ for
  notational convenience. Let $D$ be a discrete r.v.\ with support
  $\{b_1,\dots,b_{m+1}\}$ and $\Pr(D = b_i) = \frac{\Delta_i -
    \Delta_{i-1}}{U}$ for $i=1,\dots,m$, $\Pr(D = b_{m+1}) =
  \frac{U-\Delta_m}{U}$. Notice that problem \eqref{eq:2sslp} can be
  described in $O(\log U)$ bits.

  With this construction, the optimal value function $Q(x,y)$ is the
  following:
  \begin{align*}
    Q(x,y) &= \min \left\{z(D) + A : w(D) \ge y - D,
    z(D) \ge U\mathbb{E}_D[w(D)] + x - U,
    w(D) \ge 0,
    z(D) \ge 0\right\} \\
    &= \min \left\{\left(U\mathbb{E}_D[w(D)] + x - U\right)^+ + A : w(D) \ge y -D, w(D) \ge 0 \right\}\\
    &= \left( U\sum_{i=1}^{m+1} \Pr(D = b_i)  (y - b_i)^+ + x - U\right)^+ + A\\
    &= \left( \sum_{i=1}^{m} (\Delta_i - \Delta_{i-1})  (y - b_i)^+ + x - U\right)^+ + A.
  \end{align*}
  In the above expression, the last equality is due to the fact that
  for $i = m+1 $ we have $(y - b_i)^+ = 0$, so the last term
  of the summation vanishes. We claim that $Q(x,y) =
  \varphi_{r_1,\dots,r_U}(x, y)$. Clearly $Q(x,y)$ is piecewise linear
  convex, nondecreasing, and its range is $[A, A+U]$ over the domain
  $[1,U]^2$. Since its slope along the $x$ axis is $1$ whenever the
  function value is $> A$, we just need to show that $Q(U, y) =
  \varphi_{r_1,\dots,r_U}(U, y)$: equality over the rest of the domain
  then follows from the analysis carried out in the proof of
  Thm.~\ref{thm:bivarhard}. We show by induction on $k=1,\dots,m+1$
  that $Q(U, y) = \varphi_{r_1,\dots,r_U}(U, y)$ for $y \in [1,
    b_k]$. For $k=1$ this is trivial because $Q(U, 1) = A =
  \varphi_{r_1,\dots,r_U}(U, 1)$. For the induction step, we need to
  ensure that the $y$-slope of $Q(U, y)$ over the interval $[b_{k-1},
    b_k]$ matches that of $\varphi_{r_1,\dots,r_U}(U, y)$. The terms
  in the summation $\sum_{i=1}^{m+1} (\Delta_i - \Delta_{i-1}) (y -
  b_i)^+ + x - U$ for $i \ge k$ are clearly 0, whereas for $i < k$ the
  slope is $\Delta_i - \Delta_{i-1}$. So, the overall slope is
  $\sum_{i=1}^{k-1} (\Delta_i - \Delta_{i-1}) = \Delta_{k-1}$, as
  desired. This concludes the proof.
\end{proof}

We remark that the stochastic LP \eqref{eq:2sslp} does not satisfy the
conditions of our DP framework because of the term
$\left(U\mathbb{E}_D[(y-D)^+] + x - U\right)^+$ in the objective
function, which is unusual for DP. More specifically, Condition
\ref{con:functions} allows a cost function $g_t$ of the form
$\left(U(y - D)^+ + x - U\right)^+$, and our DP model (as is customary
in DP) would take the expectation of the whole expression:
$\mathbb{E}_D \left[\left(U(y - D)^+ + x - U\right)^+\right]$. In
other words, it is unclear how to formulate as a DP the cost
$\left(U\mathbb{E}_D[(y-D)^+] + x - U\right)^+$, where a piecewise
linear convex function is applied to the expectation.  In any case,
Thm.~\ref{thm:2sslphard} is indicative of the difficulties faced when
dealing with two-dimensional value functions.

\section{Constructing approximation sets with bounded-size numbers}
\label{s:size}
The FPTAS that we are about to build (Sect.~\ref{s:scheme}) relies on
solving several LPs to compute an approximation of the value function
at each stage $t=T,\dots,1$. These LPs have the property that the
input to the LPs at stage $t$ depends on data computed by the LPs at
stage $t+1$ --- the data are the domain and codomain values of a
$K$-approximation set for the value function. It is well known that
the size of the output of an LP (i.e., the optimal solution and its
value) is, in the worst case, superlinear in the size of the input
data, although still polynomial. In other words, every time an LP is
solved, the size of the output numbers may increase polynomially
compared to the size of the input numbers. We therefore need a device
to limit the growth of the number sizes, to ensure that they do not
take an exponential number of bits. As will become clearer in
Sect.~\ref{s:scheme}, to limit the growth it is sufficient to have the
ability to compute $K$-approximation sets that involve numbers that
are ``not too large''. The next result shows that this can be done
efficiently (recall from Sect.~\ref{s:definitions} that $t_{\vp}$
denotes an upper bound on the time needed to evaluate $\vp$ on a
single point in its domain).

\begin{algorithm}[tb!]
  \caption{Function {\sc ScaledCompressConvInc}$(\varphi, [A, B], K = 1+\epsilon)$.}  \label{alg:scaledcompconvinc}
  \begin{algorithmic}[1]
    \small
    \STATE Let $A' \leftarrow \frac{100 \kappa_{\varphi}}{\epsilon^2
      \varphi^{\min}} A$, $B' \leftarrow \frac{100
      \kappa_{\varphi}}{\epsilon^2 \varphi^{\min}} B$.
    \STATE \label{step:rhodef} Define
      $\rho(x) := \frac{10 \varphi(\frac{\epsilon^2
          \varphi^{\min}}{100 \kappa_{\varphi}} x )}{\epsilon
        \varphi^{\min}}$
    \STATE $W' \leftarrow \{(A', \ceil{\rho(A')})\}$
    \STATE $x \leftarrow A'$
    \WHILE{$(1+\epsilon)\rho(x) < \rho(B')$} \label{alg:sas_while_begin}
    \STATE Find a point $y' \in [A', B']$ such that $(1 +
    \frac{3}{10}\epsilon)\rho(x) \le \rho(y') \le (1+
    \frac{4}{10}\epsilon) \rho(x)$ \label{alg:sas_search_step1}
    \STATE Find a point $y'' \in [A', B']$ such that $(1 +
    \frac{6}{10}\epsilon)\rho(x) \le \rho(y'') \le (1+
    \frac{7}{10}\epsilon) \rho(x)$ \label{alg:sas_search_step2}
    \STATE Find an integer point $y \in [y', y'']$ \label{alg:sas_search_step3}
    \STATE $W' \leftarrow W' \cup \{(y, \ceil{\rho(y)})\}$ \label{alg:sas_search_step4}
    \STATE $x \leftarrow y$
    \ENDWHILE \label{alg:sas_while_end}
    \STATE $W' \leftarrow W' \cup \{(B', \ceil{\rho(B')})\}$
    \STATE $W \leftarrow \left\{\left(\frac{\epsilon^2
      \varphi^{\min}}{100 \kappa_{\varphi}} y, \frac{\epsilon
      \varphi^{\min}}{10} v^\ast\right) : \left(y, v^\ast\right) \in W'
    \right\}$
    \RETURN $W$ as an array of tuples sorted by their first coordinate
  \end{algorithmic}
\end{algorithm}

\begin{proposition} \label{prop:int_apx_set}
  Let $\varphi : [A, B] \to \R^+$ with $A, B \in \Z$ be a positive
  convex nondecreasing function with Lipschitz constant
  $\kappa_\varphi$. Then for any $K := 1 + \epsilon$,
  Alg.~\ref{alg:scaledcompconvinc} returns a canonical representation of
  a $K$-approximation set of $\varphi$ of size $O(\log_K
  \frac{\varphi^{\max}}{\varphi^{\min}})$ consisting of points that
  have domain and codomain values representable in $O(\log
  \frac{1}{\epsilon} + \log \frac{\varphi^{\max}}{\varphi^{\min}} +
  \log \kappa_\varphi + \log |A| + \log |B|)$
  space. The algorithm runs in $O(t_\varphi \log \frac{\kappa_\varphi
    (B-A)}{\epsilon \varphi^{\min}} \log_K
  \frac{\varphi^{\max}}{\varphi^{\min}})$ time.
\end{proposition}
\begin{proof}
  To simplify the exposition, we assume $\epsilon \le 1/4$. The idea
  for the proof is to construct an appropriately scaled version of
  $\varphi$ that admits a fully integer $K$-approximation set,
  i.e.\ with integer domain and codomain values. After constructing
  such an integer $K$-approximation set, we apply the inverse scaling
  to obtain a $K$-approximation set for the original function, and
  show that the size of the points is bounded. The construction that
  we follow is summarized in Alg.~\ref{alg:scaledcompconvinc}, called
  {\sc ScaledCompressConvInc}. The function $\rho$ defined therein (step
  \ref{step:rhodef}) is the scaled version of $\varphi$ mentioned
  above.

  We claim that a canonical representation $W'$ of a $K$-approximation
  set for $\rho$ that is fully integer (except at its endpoints $A',
  B'$) is constructed in time $O(t_\varphi \log \frac{\kappa_\varphi
    (B-A)}{\epsilon \varphi^{\min}} \log_K
  \frac{\varphi^{\max}}{\varphi^{\min}})$ by routine
  \textsc{ScaledCompressConvInc}. We start by showing that in every
  iteration of the ``while'' loop at lines
  \ref{alg:sas_while_begin}-\ref{alg:sas_while_end}, we find a point
  $y \in [x, B']$ and a value $v^*$ such that: (i) $\rho(y) \ge (1 +
  \frac{3}{10}\epsilon) \rho(x) $, (ii) $\rho(x) \le v^* \le (1 +
  \epsilon)\rho(x)$, and (iii) $\rho(x) \le \rho(w) \le (1 +
  \epsilon)\rho(x)$ for all $w \in [x, y]$.

  Because $\varphi$ is nondecreasing, lines \ref{alg:sas_search_step1}
  and \ref{alg:sas_search_step2} of {\sc ScaledCompressConvInc} can be
  executed by binary search. We are looking for a point within an
  interval of size at least 1 starting with an interval of size
  $\frac{100 \kappa_\varphi}{\epsilon^2 \varphi^{\min}} (B-A)$, hence
  the running time for each of the two steps is $O(t_\varphi \log
  \frac{\kappa_\varphi (B-A)}{\epsilon \varphi^{\min}})$. Line
  \ref{alg:sas_search_step3} can be performed in constant time by
  setting $y = \ceil{y'}$, as we will show next. By construction,
  $\rho(y'') - \rho(y') \ge \frac{2}{10}\epsilon
  \rho(x)$. Furthermore, $\kappa_{\rho} \le \frac{\epsilon}{10}$ by
  construction and definition of $\kappa_{\varphi}$, and $\rho(x) \ge
  \frac{10}{\epsilon} \ge 1$. Therefore, $y'' - y' \ge
  \frac{2}{10}\epsilon \frac{\rho(x)}{\kappa_{\rho}} \ge 2$, and there
  must exist at least two integer points between $y'$ and $y''$. The
  chain $y' \le \ceil{y'} < y''$ implies $(1 +
  \frac{3}{10}\epsilon)\rho(x) \le \rho(y') \le \rho(\ceil{y'}) <
  \rho(y'') \le (1 + \frac{7}{10}\epsilon)\rho(x)$, showing that $y =
  \ceil{y'}$ satisfies property (i). We claim that $\ceil{\rho(y)}$,
  used in line \ref{alg:sas_search_step4} of {\sc
    ScaledCompressConvInc}, is such that $\rho(y) \le \ceil{\rho(y)}
  \le \rho(y)(1 + \frac{1}{10}\epsilon)$. Indeed,
  \begin{equation*}
    \frac{\ceil{\rho(y)}}{\rho(y)} \le \frac{1 + \rho(y)}{\rho(y)} \le
    1 + \frac{\epsilon \varphi^{\min}}{10 \varphi(\frac{\epsilon^2
        \varphi^{\min}}{100 \kappa_{\varphi}} y )} \le 1+
    \frac{1}{10}\epsilon,
  \end{equation*}
  where the last inequality follows by definition of
  $\varphi^{\min}$. Finally, we have:
  \begin{equation*}
    \frac{\ceil{\rho(y)}}{\rho(x)} \le (1+ \frac{1}{10}\epsilon)
    \frac{\rho(y)}{\rho(x)} \le (1 + \frac{1}{10} \epsilon) (1 +
      \frac{7}{10}\epsilon) \le 1 + \epsilon,
  \end{equation*}
  showing property (ii). For the last inequality we used the assumption
  that $\epsilon \le 1/4$; larger $\epsilon$ can be handled adjusting
  lines \ref{alg:sas_search_step1}-\ref{alg:sas_search_step2} of the
  algorithm. Property (iii) now follows immediately because $\rho$ is
  convex nondecreasing.

  By our discussion above, all (domain) points in $W'$ except the two
  endpoints $A'$ and $B'$ have integer values. By properties (i), (ii)
  and (iii), the set $W'$ is a canonical representation of a
  $K$-approximation set for $\rho$ (the ceiling function at the
  endpoints $A', B'$ adds only $\frac{1}{10}\epsilon$ relative error,
  by the same argument used in showing property (ii)). Clearly, the
  ``while'' loop at lines
  \ref{alg:sas_while_begin}-\ref{alg:sas_while_end} is executed at
  most $O(\log_K \frac{\varphi^{\max}}{\varphi^{\min}})$ times because
  the value of $\rho(x)$ increases by at least a bounded below
  fraction of the multiplication factor $K$. Thus, $W'$ contains
  $O(\log_K \frac{\varphi^{\max}}{\varphi^{\min}})$ points and can be
  constructed in $O(t_\varphi \log \frac{\kappa_\varphi
    (B-A)}{\epsilon \varphi^{\min}} \log_K
  \frac{\varphi^{\max}}{\varphi^{\min}})$ time.

  The final step in {\sc ScaledCompressConvInc} consists of applying to
  $W'$ the inverse scaling that transforms $\varphi$ into $\rho$, thus
  obtaining $W$. It follows that the points in $W$ belong to the
  domain and codomain of $\varphi$. The error bound, i.e.\ the
  $K$-approximation property, is trivially satisfied because it is
  satisfied for $\rho$ and the relative error is invariant to
  multiplicative scaling. To conclude the proof, note that the size of
  the numbers in $W'$, both domain and codomain, is $O(\log
  \frac{1}{\epsilon} + \log \frac{\varphi^{\max}}{\varphi^{\min}} +
  \log \kappa_\varphi + \log |A| + \log |B|)$
  because the domain is $[A', B']$ and the codomain is
  $[\frac{10}{\epsilon}, \frac{10 \varphi^{\max}}{\epsilon
      \varphi^{\min}}]$; the multiplication factor when rescaling $W'$
  to $W$ is of the same size, which gives the desired bound.
\end{proof}

\begin{algorithm}[tb!]
  \caption{Function {\sc ScaledCompressConv}$(\varphi, [A, B], K = 1+\epsilon)$.}  \label{alg:scaledcompconv}
  \begin{algorithmic}[1]
    \small
    \STATE $x^* \leftarrow \arg\min_{x \in [A, B]} \varphi(x)$ /* We assume that a minimum of $\varphi$ can be found efficiently */
    \STATE $W_A \leftarrow$ {\sc ScaledCompressConvDec}$(\varphi, [A, x^*], K)$
    \STATE $W_B \leftarrow$ {\sc ScaledCompressConvInc}$(\varphi, [x^*, B], K)$
    \STATE $W \leftarrow (W_A \cup W_B) \setminus \{(x^*,\ceil{ \vp(x^*) } \}$
    \RETURN $W$ as an array of tuples sorted by their first coordinate
  \end{algorithmic}
\end{algorithm}

We define Function {\sc ScaledCompressConvDec} for positive convex
nonincreasing Lipschitz continuous functions in a similar way. We
define Function {\sc ScaledCompressConv} for (not necessarily
monotone) convex Lipschitz continuous functions as described in
Alg.~\ref{alg:scaledcompconv}. The algorithm assumes that a minimum of
$\varphi$ can be found exactly; in the FPTAS, this will be the case,
as the minima will be computable by solving an LP.

\begin{theorem}
  \label{thm:int_apx_set}
  Let $\varphi : [A, B] \to \R^+$ with $A, B \in \Z$ be a positive
  convex function with Lipschitz constant $\kappa_\varphi$, and assume
  that $x^* \in \arg\min_{x \in [A, B]} \varphi(x)$ can be efficiently
  computed. Then for any $K := 1 + \epsilon$,
  Alg.~\ref{alg:scaledcompconv} returns a canonical representation of
  a $K$-approximation set of $\varphi$ of size $O(\log_K
  \frac{\varphi^{\max}}{\varphi^{\min}})$ consisting of points that
  have domain and codomain values representable in $O(\log
  \frac{1}{\epsilon} + \log \frac{\varphi^{\max}}{\varphi^{\min}} +
  \log \kappa_\varphi + \log |A| + \log |B|)$ space. The algorithm
  runs in $O(t_\varphi \log \frac{\kappa_\varphi (B-A)}{\epsilon
    \varphi^{\min}} \log_K \frac{\varphi^{\max}}{\varphi^{\min}})$
  time.
\end{theorem}
\begin{proof}
  Let $\hat{\varphi}$ be the convex extension of $\varphi$ induced by
  $W$. Let $\ell := \max\{ x \in W : x \le x^*\}, r := \min\{x \in W :
  x \ge x^*\}$. Clearly $\varphi(x) \le \hat{\varphi}(x) \le K
  \varphi(x)$ for $x \in [A, \ell] \cup [r, B]$ because $W_A$ and
  $W_B$ are $K$-approximation sets over the respective domains. By
  construction (see the proof of Prop.~\ref{prop:int_apx_set}), we have
  $\varphi(\ell) \le K \varphi(x^*)$ and $\varphi(r) \le K
  \varphi(x^*)$. But then the convex extension $\hat{\varphi}$ induced
  by $W$ is such that $\varphi(x) \le \hat{\varphi}(x) \le K
  \varphi(x)$ also for $x \in [\ell, r]$. The size of the numbers
  follows directly from Prop.~\ref{prop:int_apx_set}.
\end{proof}
We remark that the main reason for discarding $(x^*,\ceil{\vp(x^*)})$
from $W$ in Alg.~\ref{alg:scaledcompconv} is to ensure that the size
of the numbers remains bounded as indicated by
Prop.~\ref{prop:int_apx_set}; in a practical implementation of the
routine, it is recommended to keep in the canonical representation
some bounded-size rounding of $(x^*,\ceil{\vp(x^*)})$.

\section{Approximating expectations with vectors of continuous random
  variables}
\label{s:exp}
To construct an approximation of the value function \eqref{eq:vf}, we
must be able to compute expected values efficiently using the given
representation of the r.v.s. Under Condition~\ref{con:implicit}, we
can compute a $K$-approximation set for the CDF of each of the
components of the vectors of r.v.s. More precisely, in case (i) of
Condition~\ref{con:implicit} we can use the algorithm given in
\cite[Sec.~8.2]{nannifptascdpfull}, and in case (ii) we can use the
algorithm given in \cite[Sec.~4.1]{halman14full} (see the 
appendix \ref{apx:functions}). Within a DP setting, the expression of
the value function contains the expectation of a composition of
functions, and we need to find a way to efficiently compute this
expectation. Under Conditions~\ref{con:implicit} and
\ref{con:functions}, at each stage the vector of random variables
$\vD_t$ is mapped to a scalar r.v.\ via a linear transformation, i.e.,
$\vec{\theta}^D$ or $\vec{\sigma}^D$. However, we do not have explicit
access to the CDF of the corresponding scalar r.v., but only to the
CDF of each component of the sum. We now show how to efficiently
construct an approximate oracle for the CDF of a sum of r.v.s, using
only oracles to their CDFs. We focus on the case of an unweighted sum
of continuous r.v.s, remarking that a similar construction for a
weighted sum of continuous or discrete r.v.s is straightforward. At
the same time, our approach can be applied to the case where instead
of having exact access to the CDFs of the r.v.s, one has an FPTAS to
it. Below, by ``generalized PDF'' we mean a probability distribution
function (PDF) that may include the Dirac delta function $\delta$ in
its definition. Recall that the Dirac delta function can be informally
thought of as a function such that $\delta(0) = \infty$, $\delta(x) =
0 \; \forall x \neq 0$ and $\int_{-\infty}^{+\infty} \delta(x) \diff x
= 1$; a formal definition is out of the scope of this paper.
\begin{proposition}
  \label{prp:convolution}
  For $i=1,\dots,\ell$, let $X_i$ be a truncated continuous r.v.\ with support
  $[X_i^{\min}, X_i^{\max}]$. Suppose $X_i$ admits a (generalized) PDF
  $f_i$ and $\kappa$-Lipschitz continuous CDF $F_i$. Assume $\gamma =
  \min_{i=1,\dots,\ell} \{\Pr(X_i = X_i^{\min}), \Pr(X_i =
  X_i^{\max})\} > 0$. Denote $X^{\min} = \min_{i=1,\dots,\ell}
  X_i^{\min}$, $X^{\max} = \max_{i=1,\dots,\ell} X_i^{\max}$, and let
  $t_F = \max_{i=1,\dots,\ell} t_{F_i}$. If $X_1,\dots,X_\ell$ are
  independent, then for any $K = 1 + \epsilon$, we can compute a
  $K$-approximation set with $O(\frac{\ell}{\epsilon} \log
  \frac{1}{\gamma})$ points for the CDF of $\sum_{i=1}^\ell X_i$, in
  $O(\frac{t_F \ell^3}{\epsilon^2} \log^2 \frac{1}{\gamma} \log
  (\kappa \ell(X^{\max} - X^{\min})))$ time.
\end{proposition}
The proof technique is based on building an approximation set of the
sum of r.v.s adding one variable at a time: assuming we have an
approximation set $W$ for the CDF of $\sum_{i=1}^{k-1} X_i$, we show
that an approximation for the CDF of $\sum_{i=1}^{k} X_i$ can be
constructed iterating over the points in $W$. 

\begin{proof}
  Let $\bar{K} := \sqrt[\ell]{1+\epsilon}$. We proceed by induction on
  $k=1,\dots,\ell$, computing a $\bar{K}^k$-approximation set for the
  CDF of $\sum_{i=1}^k X_i$ that contains $O(\frac{\ell}{\epsilon}
  \log \frac{1}{\gamma})$ points and requires $O(\frac{t_F
    \ell^2}{\epsilon^2} \log^2 \frac{1}{\gamma} \log (\kappa
  k(X^{\max} - X^{\min})))$ time at each induction step. In these
  expressions, the factor $\frac{\ell}{\epsilon}$ comes from the fact
  that $\log_{\bar{K}}(\cdot)$ is $O(\frac{\ell}{\epsilon} \log
  (\cdot))$ because $K = \sqrt[\ell]{1+\epsilon}$.  For $k=1$, we
  simply apply {\sc CompressInc}, defined in \cite{halman14full} if
  $D_t$ is a discrete random variable, and in \cite{nannifptascdpfull}
  if $D_t$ is a continuous random variable, see the appendix
  \ref{apx:functions}.  This yields a $\bar{K}$-approximation set
  $W_1$ of $X_1$ with $O(\frac{\ell}{\epsilon} \log \frac{1}{\gamma})$
  points in $O(\frac{\ell t_F}{\epsilon} \log \frac{1}{\gamma} \log
  (\kappa (X^{\max} - X^{\min})))$ time.

  Denote $Y = \sum_{i=1}^{k-1} X_i$. Now assume that we have a
  $\bar{K}^{k-1}$-approximation set $W_{k-1} := \{b_1,\dots,b_m\}$
  for the CDF $F_Y$ of $Y$, and let $\tilde{F}_Y$ be the corresponding
  monotone extension. Recall that $\tilde{F}_Y$ can be seen as a
  summation of step functions. Let $h(x)$ be the step function $h(x) =
  1$ if $x \ge 0$, $h(x) = 0$ otherwise. Using this notation, we can
  express $\tilde{F}_Y(x) = \sum_{i=1}^m a_i h(x - b_i)$ for some $a_i
  \in \R^+$.

  We want to compute $\varphi(z) := \Pr(X_k + Y \le z)$ for any given
  $z$. Because $X_k$ and $Y$ are independent, we can write:
  \begin{align*}
    \int_{\R} \int_{-\infty}^{z-x} f_{X_k,Y}(x,y)\diff y \diff x &&=&&
    \int_{\R} F_Y(z-x) f_k(x) \diff x \le \\
    \int_{\R} \tilde{F}_Y(z-x) f_k(x) \diff x &&\le&&
    \bar{K}^{k-1} \int_{\R} F_Y(z-x) f_k(x) \diff x,
  \end{align*}
  where the inequalities follow from the facts that $\tilde{F}_Y$ is a
  $\bar{K}^{k-1}$-approximation of $F_Y$, and $f_k$ is nonnegative.
  This shows that $\tilde{\varphi}(z) := \int_{\R} \tilde{F}_Y(z-x)
  f_k(x) \diff x$ is a $\bar{K}^{k-1}$-approximation of $\varphi(z)$,
  i.e.\ the CDF of $X_k + Y$. Breaking down $\tilde{F}_Y$ into a
  summation of step functions, we can expand this integral as:
  \begin{align*}
    \int_{\R} \sum_{i=1}^m a_i h(z-x-b_i) f_k(x) \diff x
    = \sum_{i=1}^m a_i \int_{-\infty}^{z-b_i} f_k(x) \diff x =
    \sum_{i=1}^m a_i \Pr(X_k \le z - b_i).
  \end{align*}
  It follows that $\tilde{\varphi}(z)$ can be computed in time
  $O(mt_{F_k})$. By the induction hypothesis, this is $O(\frac{t_F
    \ell}{\epsilon} \log \frac{1}{\gamma})$. We now apply {\sc
    CompressInc} to $\tilde{\varphi}$ with approximation factor
  $\bar{K}$. By approximation of approximation
  (Prop.~\ref{prp:CAF}(7)), this yields an approximation set $W_k$
  that induces a $\bar{K}^k$-approximation function of
  $\varphi$. $W_k$ has $O(\frac{\ell}{\epsilon} \log
  \frac{1}{\gamma})$ points, and the computation takes $O(\frac{t_F
    \ell^2}{\epsilon^2} \log^2 \frac{1}{\gamma} \log (\kappa
  k(X^{\max} - X^{\min})))$ time.  This concludes the inductive
  claim.

  To conclude the proof, we note that the loop discussed above is
  repeated for $k=1,\dots,\ell$, therefore the total runtime is
  $O(\frac{t_F \ell^3}{\epsilon^2} \log^2 \frac{1}{\gamma} \log
  (\kappa \ell(X^{\max} - X^{\min})))$.
\end{proof}

Prop.~\ref{prp:convolution} is given for truncated continuous r.v.s,
but it is straightforward to note that the same approach works for
discrete r.v.s. In this case, instead of the PDF of the r.v.s we
simply employ the corresponding PMF, obtaining the same
result. Furthermore, we note that the proof can be adapted to the case
in which only approximate access to the CDF of the r.v.s is
available. Indeed, suppose we have an FPTAS for the CDF of each the
r.v.s. In this case, we simply need to decrease the approximation
factor $\bar{K}$ used in the proof to, say, $\bar{K} =
\sqrt[2\ell]{1+\epsilon}$, and replace each call to the oracles for
the CDF to a $\sqrt[2\ell]{1+\epsilon}$-approximate call to the FPTAS
for the CDF. In this respect we improve upon
\cite{li14convolution}. The running time given by
\cite{li14convolution} to compute $\Pr(\sum_{i=1}^\ell X_i \le \rho)$
for discrete integer r.v.s is $O(\frac{t_F \ell^3}{\epsilon^2} \log^2
\frac{1}{\gamma} \log \rho)$, essentially the same as ours if we
restrict our algorithm to the discrete case: we have $\log \ell
(X^{\max} - X^{\min})$ rather than $\log \rho$ because we compute the
CDF for all values of $\rho$ simultaneously.

\begin{algorithm}[tb!]
  \caption{Function {\sc CompressConvolution}$((X_1,\dots,X_\ell), (c_1,\dots,c_\ell), K)$.}  \label{alg:compressconvolution}
  \begin{algorithmic}[1]
    \small
    \STATE $\bar{K} \leftarrow \sqrt[\ell]{K}$
    \STATE $W_1 \leftarrow$ {\sc CompressInc}$(F_1(\frac{\cdot}{c_1}), [X_1^{\min}, X_1^{\max}], \bar{K})$ \quad  /* $F_1(\cdot)$ is the CDF of $X_1$ */
    \FOR{$k=2,\dots,\ell$}
    \STATE Denote $\{b_1,\dots,b_m\}$ the points in the approximation set $W_{k-1}$
    \STATE Let $a_1,\dots,a_m$ be the coefficients of the monotone extension $\tilde{F}_{k-1}(x) = \sum_{i=1}^m a_i h(x - b_i)$ of $W_{k-1}$
    \\ /* $h(x)$ is the unit step function: $h(x) = 1$ for $x \ge 0$, $0$ otherwise */
    \\ /* $\tilde F_{k-1}(\cdot)$ is a $\bar K^{k-1}$-approximation of the CDF of $\sum_{j=1}^{k-1} c_j X_j$ */
    \STATE Define $\tilde{\varphi}(z) := \sum_{i=1}^m a_i F_k(\frac{z - b_i}{c_k})$ \quad /* $F_k(\cdot)$ is the CDF of $X_k$ */
    \STATE $W_k \leftarrow$ {\sc CompressInc}$(\tilde{\varphi}, [\sum_{j=1}^k X_j^{\min}, \sum_{j=1}^k X_j^{\max}], \bar{K})$
    \ENDFOR
    \RETURN $W_{\ell}$ as an array of tuples sorted by their first coordinate
  \end{algorithmic}
\end{algorithm}

We define in Alg.~\ref{alg:compressconvolution} a routine {\sc
  CompressConvolution}$(\vec{X}, \vec{c}, K)$ that implements the
constructive proof of Prop.~\ref{prp:convolution}. The routine takes
as input a vector of random variables $\vec{X}$ and returns a
canonical representation of a $K$-approximation oracle function for
$\varphi(z) := \Pr(\vec{c} \cdot \vec{X} \le z)$. The query time of
the oracle is $O\big(\log (\frac{\ell}{\eps} \log
\frac{1}{\gamma})\big)$.

We can now assume that we have approximate access to the CDF of the
random variables defined by $\vec{\theta}^D \cdot \vD_t$ and
$\vec{\sigma}^D \cdot \vD_t$ that appear in the cost functions, as
detailed in Condition~\ref{con:functions}. We then need to compute the
expectation of a composition of functions using the approximate
CDF. If $\vec{\theta}^D \cdot \vD_t$ and $\vec{\sigma}^D \cdot \vD_t$
have a discrete distribution, i.e.\ case (ii) of
Condition~\ref{con:implicit}, we can rely on
\cite[Prop.~6.8]{halman17sample}. The rest of this section details
how to efficiently compute expectations involving continuous r.v.s,
thereby taking care of case (i) of Condition~\ref{con:implicit}.

Our next result shows that we can efficiently compute an expected
value of a composition of functions involving a continuous
r.v. Essentially, the proposition is the analog of
\cite[Prop.~6.8]{halman17sample} for continuous r.v.s, but the proof
technique has some significant differences, because for continuous
r.v.s we do not know how to write an easy-to-evaluate closed-form
expression for the expected value.

\begin{proposition} \label{prp:implicitexp}
Let $\xi:[A,B] \to \R^+$ be a (not necessarily monotone) convex
function. Let $K_1,K_2 \geq 1$. Let $\psi:[A,B] \to \R^+$ be an
increasing piecewise linear convex function with breakpoints $A =
a_1<\ldots<a_n <B$ and slopes $0=\Delta_0 \leq \Delta_1 < \ldots <
\Delta_n$ that $K_1$-approximates $\xi$. Let $D$ be a (not necessarily
nonnegative) truncated continuous r.v.\ with support $[D^{\min},
  D^{\max}]$ and such that $\min\{\Pr(D = D^{\min}), \Pr(D =
D^{\max})\} > 0$. Suppose $D$ admits a (generalized) PDF $F'$, and let
$F$ be the corresponding CDF. Let $\tilde F(\cdot)$ be a monotone
$K_2$-approximation of $F$ with breakpoints at $D^{\min} = d_1 <
\ldots < d_m = D^{\max}$, and let $\tilde F(d_0)=0$. Let
$f(x,d)=bx+e-d$ for some given $b,e \in \R$, and $m_i(x) = \max \{j
\;|\; d_j \leq bx+e-a_i\}$. Then
\begin{equation}
  \tilde \xi(x) =
\psi(A)+\sum_{i=1}^n (\Delta_i-\Delta_{i-1}) \left(
(bx+e-d_{m_i(x)}-a_i) \tilde F(d_{m_i(x)+1})+ \sum_{k=1}^{m_i(x)-1}
(d_{k+1}-d_k) \tilde F(d_{k+1}) \right)
\label{eq:tildexi}
\end{equation}
is a convex piecewise linear $K_1K_2$-approximation function of
$\mathbb{E}_D(\xi(f(x,D)))$ with $O(mn)$ pieces.  Furthermore, one can
construct in $O(mn\log n)$ time an oracle for $\tilde \xi(\cdot)$ in
the form of either a canonical representation, or a representation
that consists of breakpoints and slopes, each of size $O(mn)$.
\end{proposition}
The proof technique involves decomposing
$\mathbb{E}_{D_t}(\xi(f(x,D)))$ as a finite sum involving $\tilde{F}$
for each piece of the piecewise linear function $\varphi$. The error
bound is obtained by constructing ``easy-to-handle'' discrete random
variables $\hat{D}, \check{D}$ such that $\hat{D} \preceq D_t \preceq
\check{D}$ in the usual stochastic order. Since the proof is long and
technical, it is given in appendix \ref{apx:proofimplicitexp}.

As in \cite{halman17sample}, it is easy to see that
Prop.~\ref{prp:implicitexp} can be applied to separable linear
transition functions $f(x,d)=bx+e+cd$ for any coefficient $c$, using a
transformed r.v.\ $D'$ with $D':=D/(-c)$, and a similar result holds
for convex functions approximated by a decreasing piecewise linear
convex function. As a consequence, we obtain a modified version of
\cite[Cor.~6.9]{halman17sample}, stated below.

\begin{proposition} \label{prop:exp}
Let $\xi:[A,B] \to \R^+$ be a convex function. Let $K_1, K_2 \geq
1$. Let $\psi:[A,B] \to \R^+$ be a piecewise linear convex function
with $n$ breakpoints that $K_1$-approximates $\xi$ over $[A,B]$. Let
$D$ be a (not necessarily nonnegative) continuous r.v.\ satisfying the
conditions stated in Prop.~\ref{prp:implicitexp}. Let $\tilde
F(\cdot)$ be a monotone nondecreasing $K_2$-approximation of the CDF
of $D$ with $m$ breakpoints. Let $f(x,D) = b x + e + cd$, with $b, e,
c \in \R$. Then we can construct in $O(mn\log n)$ time a convex
piecewise linear $K_1K_2$-approximation function of
$\mathbb{E}_D(\xi(f(x,D)))$ in the form of either a canonical
representation, or a representation that consists of breakpoints and
slopes, each of size $O(mn)$.
\end{proposition}
The proof of Prop.~\ref{prop:exp} follows from an application of
Prop.~\ref{prp:implicitexp} twice, to the nondecreasing and
nonincreasing part of $\psi$. In Alg.~\ref{alg:compressexpval} we
define a routine {\sc CompressExpVal}$(\psi, (b,e,c), \tilde{F})$ that
takes as input $\psi, b, e, c, \tilde{F}$ as defined in
Prop.~\ref{prop:exp}, and returns an oracle for a
$K_1K_2$-approximation function of $\mathbb{E}_D(\xi(f(\cdot,D)))$ in
the form of a canonical representation. The query time for the oracle
is $O(\log mn)$. {\sc CompressExpVal} uses one subroutine for the
increasing part of the convex function $\psi$, called {\sc
  CompressExpValInc} and defined in Alg.~\ref{alg:compressexpvalinc},
and one for the decreasing part, called {\sc CompressExpValDec}. The
pseudocode for {\sc CompressExpValDec} is the same as {\sc
  CompressExpValInc} after ``mirroring'' the function $\psi$ given as
the first argument to transform it from decreasing to increasing.

\begin{algorithm}[tb!]
  \caption{Function {\sc CompressExpVal}$(\psi, (b, e, c), \tilde{F})$.}  \label{alg:compressexpval}
  \begin{algorithmic}[1]
    \small
    \STATE Let $[A, B]$ be the domain of $\psi$ and $x^*$ one of its minimizers
    \STATE Define $\psi_1$ as the restriction of $\psi$ to $[A, x^*]$, $\psi_2$ as the restriction of $\psi$ to $[x^*, B]$
    \STATE $W_1 \leftarrow$ {\sc CompressExpValDec}$(\psi_1, (b, e, c), \tilde{F})$
    \STATE $W_2 \leftarrow$ {\sc CompressExpValInc}$(\psi_2, (b, e, c), \tilde{F})$
    \RETURN $W_1 \cup W_2$ as an array of tuples sorted by their first coordinate
  \end{algorithmic}
\end{algorithm}

\begin{algorithm}[tb!]
  \caption{Function {\sc CompressExpValInc}$(\psi, (b, e, c), \tilde{F})$.}  \label{alg:compressexpvalinc}
  \begin{algorithmic}[1]
    \small
    \STATE Using the approximation set for r.v.~$D$ inducing $\tilde{F}$, compute an approximation set $W_{D'}$ for the transformed random variable $D' = -D/c$
    \STATE Let $\tilde{F}'$ be the monotone extension of $W_{D'}$, $d_1,\dots,d_m$ be the breakpoints of $W_{D'}$
    \STATE Let $a_1,\dots,a_n$ be the breakpoints of $\psi$, $\Delta_0,\dots,\Delta_n$ its slopes
    \STATE $S \leftarrow \text{sort}\left(\bigcup_{i=1}^n \{d_j -e + a_i : j = 1,\dots,m\}\right)$
    \STATE Compute the $m-1$ partial sums $\sum_{k=1}^{j}(d_{k+1}-d_k) \tilde F'(d_{k+1})$ for $j=1,\dots,m-1$
    \STATE $W \leftarrow \{(x, \tilde{\xi}(x) : x \in S\}$, where $\tilde{\xi}$ is defined in \eqref{eq:tildexi}
    \STATE Loop over the elements of $W$ to eliminate redundant breakpoints where the slope does not change
    \RETURN $W$ as an array of tuples sorted by their first coordinate
  \end{algorithmic}
\end{algorithm}

\section{The approximation scheme}
\label{s:scheme}
We now show how to approximate a DP that satisfies Conditions
\ref{con:sets}-\ref{con:functions}. We give a full proof of the
correctness of the algorithm for the simpler Condition
\ref{con:functions}(i), according to which the cost functions are
known explicitly. The FPTAS under Condition \ref{con:functions}(ii) is
similar, and its proof only highlights the differencs from Condition
\ref{con:functions}(i).

\subsection{Condition~\ref{con:functions}(i): explicit functions}
\label{s:schemei}

Recall that by Condition \ref{con:functions}(i) we have $g_t(I_t,
\vx_t, \vD_t) = g_t^I(I_t,\vx_t) + g_t^D(f^g_t(I_t,\vx_t,\vD_t))$; to
avoid additional notation, we will not specify the expression of
$g_t^I$ and $g_t^D$ in terms of breakpoints and slopes, but we will
exploit their structure to construct an FPTAS. (Hereafter, given any
function $f$ we use the notation $f(\cdot)$ where ``$\cdot$'' stands
for the varying argument of $f$.)

\begin{proposition}
  \label{prop:vflp}
  Suppose the DP formulation \eqref{eq:dp} satisfies Conditions
  \ref{con:sets}-\ref{con:functions}. Let $K_1, K_2, K_3 > 1$. Let
  $\hat{z}_{t+1}$ be a piecewise linear convex $K_1$-approximation of
  the value function $z_{t+1}$. For a given $I_t \in S_t$, let
  $\tilde{Z}_{t+1}(\cdot)$ be a piecewise linear convex
  $K_2$-approximation of $\mathbb{E}[\hat{z}_{t+1}(\cdot + \theta^I
    I_t + \vec{\theta}^D \cdot \vD_t)]$, and $\tilde{G}_t^D(\cdot)$ be
  a piecewise linear convex $K_3$-approximation of
  $\mathbb{E}[g_t^D(\cdot + \sigma^I I_t + \vec{\sigma}^D \cdot
    \vD_t)]$. Then the value of the following mathematical program:
  \begin{equation}
    \left.
    \begin{array}{rrcl}
      \displaystyle \bar{z}_t(I_t) := \min_{\vx_t} & g_t^I(I_t, \vx_t) +
      \tilde{G}_t^D(w) + \tilde{Z}_{t+1}(y) & & \\
      & A_t \vx_t &\ge& \vec{b}_t + \vec{\delta}_{b_t} I_t \\
      & \vec{\sigma}^x \cdot \vx_t - w &=& 0 \\
      & \vec{\theta}^x \cdot \vx_t - y &=& 0 \\
      &\vx_t &\ge& 0,
    \end{array}
    \right\}
    \label{eq:vflp}
  \end{equation}
  is a $\max \{K_1K_2, K_3\}$-approximation of the optimal value
  function $z_t(I_t)$. \eqref{eq:vflp} can be cast as an LP using the
  standard $\min\max$ reformulation of piecewise linear convex
  functions.
\end{proposition}
\begin{proof}
  The expression for the optimal value function is:
  $$z_t(I_t)=\min_{\vx_t \in \A_t(I_t)} \mathbb{E} \{
  g_t(I_t,\vx_t,\vD_t) \allowbreak + z_{t+1}(f_t(I_t,\vx_t,\vD_t))
  \}.$$ By linearity of the expectation and by definition of $g_t^I,
  \tilde{G}_t^D, \tilde{Z}_{t+1}, f_t, f_t^g$ and $A_t(I_t)$, it is
  straightforward to check that \eqref{eq:vflp} corresponds exactly to
  the expression of the value function, with two of the terms replaced
  by approximations. The fact that the resulting value is a $\max
  \{K_1K_2, K_3\}$ approximation of $z_t(I_t)$ follows from
  Prop.~\ref{prp:CAF} (summation of approximation and approximation of
  approximation).
\end{proof}

\begin{algorithm}[tb!]
  \caption{Procedure {\sc ApxScheme1}$(\epsilon)$ for Condition~\ref{con:functions}(i)}
  \label{alg:fptas1}
  \begin{algorithmic}[1]
    \small
    \STATE $K \leftarrow \sqrt[2T]{1+\epsilon}, \; \hat{z}_{T+1} \leftarrow g_{T+1}$
    \FOR{$t:=T$ {\bf downto} 1}
    \STATE \label{step:convolution1} $\tilde{F}_g \leftarrow$ {\sc
  CompressConvolution}$(\vD_t, \vec{\sigma}^D, K)$
    \STATE \label{step:convolution2} $\tilde{F}_z \leftarrow$ {\sc
  CompressConvolution}$(\vD_t, \vec{\theta}^D, K)$
    \STATE \label{step:gtilde} For fixed $I_t$, define $\tilde{G}_t^D \leftarrow$ {\sc CompressExpVal}$(g_t^D, (1, \sigma^I I_t, 1) , \tilde{F}_g)$ \\ /* $\tilde G_t^D(\cdot)$ is an oracle for a $K$-approximation of $\mathbb{E}[g^D_t(\cdot+\sigma^I I_t+\vec{\sigma}^D \cdot \vD_t)]$ */
    \STATE \label{step:ztilde} For fixed $I_t$, define $ \tilde{Z}_{t+1} \leftarrow$ {\sc CompressExpVal}$(\hat{z}_{t+1}, (1, \theta^I I_t, 1), \tilde{F}_z)$ \\ /* $\tilde Z_{t+1}(\cdot)$ is an oracle for a $K$-approximation of $\mathbb{E}[\hat{z}_{t+1}(\cdot+\theta^I I_t+\vec{\theta}^D \cdot \vD_t)]$ */
    \STATE \label{step:zhat} $\hat{z}_t \leftarrow$ {\sc ScaledCompressConv}$(\bar{z}_t, [\S_{t}^{\min}, \S_{t}^{\max}], K)$, where $\bar{z}_t$ is defined as in \eqref{eq:vflp}
    \ENDFOR \RETURN $\hat{z}_1$
  \end{algorithmic}
\end{algorithm}


We are now ready to state our main result. The proof consists of an
application of the different results already stated in this paper,
being careful to keep track of the runtime of the algorithm.

\begin{theorem} \label{thm:FPTAS1}
  Given a stochastic DP satisfying Conditions
  \ref{con:sets}-\ref{con:functions}(i), {\sc ApxScheme1}($\epsilon$)
  (Algorithm~\ref{alg:fptas1}) computes a $(1+\epsilon)$-approximation
  of the optimal value function $z_1$, and runs in polynomial time in
  the binary input size and $1/\epsilon$.
\end{theorem}
\begin{proof}
  The proof proceeds by induction for $t=T+1,\dots,1$, showing that at
  stage $t$ we obtain a piecewise linear convex
  $K^{2(T+1-t)}$-approximation $\hat{z}_t$ of the value function $z_t$
  via an approximation set of cardinality $O(\log_K
  \frac{(T+2-t)U_g}{g^{\min}_{T+1}}) = O(\frac{T}{\epsilon} \log
  \frac{(T+2-t)U_g}{g^{\min}_{T+1}})$ since $K =
  \sqrt[2T]{1+\epsilon}$. The argument of the $\log$ follows from the
  fact that the largest value of the value function at stage $t$ is
  bounded above by $(T+2-t)U_g$, and its smallest value is bounded
  below by $g^{\min}_{T+1}$. The step $t=T+1$ is trivial, as we have
  $\hat{z}_{T+1} = g_{T+1}$, which is exactly the value function at
  the terminal stage, and we are given its piecewise linear convex
  description in the input. We now show the induction step.

  At step \ref{step:convolution1}, using the
  routine {\sc CompressConvolution}, we compute a $K$-approximation
  $\tilde{F}_g$ of the CDF of $\vec{\sigma}^D \cdot \vD_t$ in the form of a canonical representation. By Prop.~\ref{prp:convolution},
  the cardinality of the canonical representation is $O(\frac{T
    \ell}{\epsilon} \log \frac{1}{\gamma})$, and this step requires
  $O(\frac{t_F T \ell^3}{\epsilon^2} \log^2 \frac{1}{\gamma} \log
  (\kappa \ell n^* U_f))$ time. Step \ref{step:convolution2} is
  similar and computes a $K$-approximation $\tilde{F}_z$ of the CDF of
  $\vec{\theta}^D \cdot \vD_t$.

  At step \ref{step:gtilde}, for a fixed value of $I_t$ we define a
  $K$-approximation $\tilde{G}_t^D(\cdot)$ of
  $\mathbb{E}_{\vD_t}[g_t^D(\cdot + \sigma^I I_t + \vec{\sigma}^D
    \cdot \vD_t)]$, using Prop.~\ref{prop:exp} with parameters set to
  $\xi=\psi=g_t^D, \; n=O(m_t), \; m=O(\frac{T \ell}{\epsilon} \log
  \frac{1}{\gamma}), \; K_1=1, \; K_2=K$. This is possible because by
  Condition~\ref{con:functions}(i) we have a description of $g_t^D$ in
  terms of breakpoints and slopes, which is equivalent to a
  $1$-approximation set of $g_t^D$ with $O(m_t)$ points. We remark
  that no actual computation is involved in this step because we did
  not fix a value of $I_t$ yet (it will be determined later in step
  \ref{step:zhat}), but including step \ref{step:gtilde} in the
  algorithm helps us for the analysis. By Prop.~\ref{prop:exp}, given
  a value for $I_t$, computing $\tilde{G}_t^D$ in the form of a
  canonical representation of size $O(\frac{m_t T \ell}{\epsilon} \log
  \frac{1}{\gamma})$, requires $O(\frac{m_t T \ell}{\epsilon} \log
  \frac{1}{\gamma} \log m_t)$ time.

  At step \ref{step:ztilde}, for a fixed value of $I_t$ we define a
  $K$-approximation $\tilde{Z}_{t+1}(\cdot)$ of $\mathbb{E}[\hat
    z_{t+1}(\cdot + \theta^I I_t + \vec{\theta}^D \cdot \vD_t)]$,
  using Prop.~\ref{prop:exp} with parameters set to $\xi=\psi=\hat
  z_{t+1}, \; n=O(\frac{T}{\epsilon} \log
  \frac{(T+1-t)U_g}{g^{\min}_{T+1}}), \; m=O(\frac{T\ell}{\epsilon}
  \log \frac{1}{\gamma}), \; K_1=1$ and $K_2=K$. Given a value for
  $I_t$, computing $\tilde{Z}_{t+1}$ in the form of a canonical representation of size $O(\frac{T^2
    \ell}{\epsilon^2} \log \frac{(T+1-t)U_g}{g^{\min}_{T+1}}
  \log\frac{1}{\gamma})$, requires $O(\frac{T^2 \ell}{\epsilon^2} \log
  \frac{(T+1-t)U_g}{g^{\min}_{T+1}} \log\frac{1}{\gamma} \log
  \frac{T}{\epsilon} \log \log \frac{(T+1-t)U_g}{g^{\min}_{T+1}})$
  time.  

  The main workhorse of the FPTAS is step \ref{step:zhat}, that
  computes an approximation $\hat z_t$ for the value function $z_t$, using
  the function $\bar{z}_t$ as defined in \eqref{eq:vflp}. Note that by
  Prop.~\ref{prop:vflp} with parameters set to $K_1=K^{2(T-t)}$ and
  $K_2=K_3=K$, $\bar{z}_t$ is $K^{2(T-t)+1}$-approximation of
  $z_t$. The routine {\sc ScaledCompressConv} can be used because a
  minimum of $\bar{z}_t$ can be found exactly with the solution of a
  single LP. Applying Thm.~\ref{thm:int_apx_set} to $\bar{z}_t$,
  coupled with Def.~\ref{def:apxset} and approximation of
  approximation (Prop.~\ref{prp:CAF}(7))) we obtain a
  $K^{2(T+1-t)}$-approximation $\hat z_t$ of $z_t$ in the form of a canonical representation of size
  $O(\frac{T}{\epsilon} \log \frac{(T+2-t)U_g}{g^{\min}_{T+1}})$. This
  concludes the induction claim. It remains to show that the runtime of this step is polynomial. By Thm.~\ref{thm:int_apx_set}, the
  runtime of step \ref{step:zhat} is $O(t_{\bar{z}} \log
  \frac{\kappa_{\bar{z}_t} U_S}{\epsilon g^{\min}_{T+1}} \log_K
  \frac{(T+2-t)U_g}{g^{\min}_{T+1}})$. Note that by \eqref{eq:vflp}
  the Lipschitz constant $\kappa_{\bar{z}_t}$ is bounded above by
  $(3U_f U_A^2)^{(T+1-t)} \kappa$. We now proceed to bound
  $t_{\bar{z}}$, i.e., the time to construct and compute
  \eqref{eq:vflp} at a given value of $I_t$.

  We first analyze the size of the LP \eqref{eq:vflp}. The variables
  $y,w$ used in \eqref{eq:vflp} can be substituted out. Thus,
  \eqref{eq:vflp} has $p + 3$ variables: $p$ variables in the
  description of $\A_t(I_t)$, plus one variable for the minmax
  formulation of each of the convex functions $g_t^I(I_t,\vec{\cdot}),
  \; \tilde{G}_t^D(\cdot)$ and $\tilde{Z}_{t+1}(\cdot)$. The LP has
  $O(m + q_t + \frac{m_t T \ell}{\epsilon} \log \frac{1}{\gamma} +
  \frac{T^2 \ell}{\epsilon^2} \log \frac{(T+1-t)U_g}{g^{\min}_{T+1}}
  \log\frac{1}{\gamma})$ constraints plus $p$ nonnegativity
  constraints: $m$ constraints come from $\A_t(I_t)$, $q_t$ from the
  minmax formulation of the piecewise linear convex cost $g_t^I$, the
  remaining constraints come from the minmax formulation of
  $\tilde{G}_t^D$ and $\tilde{Z}_{t+1}$, whose number of pieces is
  discussed above. The size of the numbers in the LP is bounded by
  $O(\max\{U_A, U_f, \log \frac{1}{\epsilon} + \log
  \frac{(T+2-t)U_g}{g^{\min}_{T+1}}\})$. An LP with these
  characteristics is solved in $O((m + q_t + \frac{m_t T
    \ell}{\epsilon} \log \frac{1}{\gamma} + \frac{T^2
    \ell}{\epsilon^2} \log \frac{(T+1-t)U_g}{g^{\min}_{T+1}} \log
  \frac{1}{\gamma} + p)^{1.5}p^2\max\{U_A, U_f, \log
  \frac{1}{\epsilon} + \log \frac{(T+2-t)U_g}{g^{\min}_{T+1}}\})$
  arithmetic operations, following \cite{renegar88}. In addition, the
  piecewise linear representation for the functions $\tilde{G}_t^D$
  and $\tilde{Z}_{t+1}$ is computed for each fixed $I_t$, which
  requires additional $O(\frac{m_t T \ell}{\epsilon} \log
  \frac{1}{\gamma} \log m_t + \frac{T^2 \ell}{\epsilon^2} \log
  \frac{(T+1-t)U_g}{g^{\min}_{T+1}} \allowbreak \log\frac{1}{\gamma}
  \log \frac{T}{\epsilon} \log \log
  \frac{(T+1-t)U_g}{g^{\min}_{T+1}})$ time.

  Hence, the total runtime to compute a single value of $\bar z_t$ is
  \begin{align}
    t_{\bar z_t}=O\bigg( \Big( \big(m + q_t + \frac{m_t T \ell}{\epsilon}
    \log \frac{1}{\gamma} + \frac{T^2 \ell}{\epsilon^2} \log
    \frac{(T+1-t)U_g}{g^{\min}_{T+1}} \log \frac{1}{\gamma} +
    p\big)^{1.5}p^2 \notag \\ \max\Big\{U_A, U_f, \log \frac{1}{\epsilon} +
    \log \frac{(T+2-t)U_g}{g^{\min}_{T+1}}\Big\}\Big) \label{eq:tbarz} \\
    + \frac{m_t T \ell}{\epsilon} \log \frac{1}{\gamma} \log m_t +
    \frac{T^2 \ell}{\epsilon^2} \log \frac{(T+1-t)U_g}{g^{\min}_{T+1}}
    \log\frac{1}{\gamma} \log \frac{T}{\epsilon} \log \log
    \frac{(T+1-t)U_g}{g^{\min}_{T+1}} \bigg). \notag
  \end{align}

  Step \ref{step:zhat} has to be repeated $T$ times. Adding the time
  to compute $\tilde{F}_z$ and $\tilde{F}_g$, we obtain the total
  runtime of the approximation scheme
  (Algorithm~\ref{alg:fptas1}):
  \begin{equation*}
    O\left( t_{\bar{z}_1}
    \frac{T^2}{\epsilon} \log \frac{\kappa_{\bar z_1} U_S}{\epsilon
      g^{\min}_{T+1}} \log \frac{(T+1)U_g}{g^{\min}_{T+1}} +
    \frac{t_F T^2
      \ell^3}{\epsilon^2} \log^2 \frac{1}{\gamma} \log (\kappa \ell
    n^* U_f) \right),
  \end{equation*}
  where $t_{\bar{z}_0}$ is as in \eqref{eq:tbarz}. If we perform the
  substitution, the dependency on $\epsilon$ is at most
  $\frac{1}{\epsilon^5}\log^2 \frac{1}{\epsilon}$, the dependency on
  $T$ is at most $T^6 \log^3 T$, the dependency on $\ell$ is at most
  $\ell^3$, and the dependency on $p$ is at most $p^{3.5}$. If
  $\kappa, U_A$ and $U_f$ are small enough so that the Lipschitz
  constant $\kappa_{\bar z_1}$ is bounded by a constant, independent
  of the problem instance, then the dependency on $T$ decreases to at
  most $T^5 \log^3 T$.
\end{proof}

\subsection{Condition~\ref{con:functions}(ii): oracle model}
\label{s:schemeii}

We provide an FPTAS via a reduction to the FPTAS for
Condition~\ref{con:functions}(i). To adapt the FPTAS to the case in
which $g_{T+1}, g_t^D$ are Lipschitz continuous univariate convex
functions accessed by value oracles, we want to construct piecewise
linear convex approximations of these functions with bounded error.
In general, relative approximations for Lipschitz continuous
univariate convex functions that are described via value oracles may
not be computable in finite time,
cfr.~\cite[Cor.~2.5]{nannifptascdpfull}. However, approximations that
attain small multiplicative plus additive error (i.e., with both
additive and relative error) can be found in polynomial time
\cite[Prop.~3.7]{nannifptascdpfull}. The optimum of \eqref{eq:dp} has
value of at least $g^{\min}_{T+1} > 0$ by assumption, therefore a
small additive error in each of the $T$ backward recursion steps can
be absorbed without forsaking the approximation guarantee. In
particular, we slightly reduce the relative approximation factor at
each stage, while allowing an additive error of the order of $\epsilon
g^{\min}_{T+1}/2(2T+1)$.

We now give more details. We start by formally defining
$(\Sigma,\Pi)$-approximation functions, as introduced in
\cite[Def.~2.10]{nannifptascdpfull}. $\Sigma$ denotes an additive
approximation error, $\Pi = 1+\epsilon$ a multiplicative approximation
error; $\Pi$ has the same role as $K$ in this paper, and after the
main definition we revert to our usual notation using $K$ for relative
error.
\begin{definition}
  \label{def:spapxfun}
  Let $\epsilon > 0, \Sigma > 0, \Pi = 1 + \epsilon$ and let $\varphi
  : [A,B] \to \R^+$. We say that $\tilde{\varphi} : [A,B] \to \R^+$
  is a {\em $(\Sigma,\Pi)$-approximation function} of $\varphi$ (or
  more briefly, a $(\Sigma,\Pi)$-approximation of $\varphi$) if for
  all $x \in [A,B]$ we have $\varphi(x) \le \tilde{\varphi}(x) \le
  \Pi\varphi(x) + \Sigma$.
\end{definition}
It is possible to efficiently construct piecewise linear convex
$(\Sigma,\Pi)$-approximation functions for Lipschitz continuous
univariate convex functions that are described via value oracles. For
space reasons, we do not give full details here, but we refer the
reader to \cite{nannifptascdpfull} or the appendix \ref{apx:sigpi}.
The next three results are generalizations of
Prop.~\ref{prp:implicitexp}, Prop.~\ref{prop:exp} and
Prop.~\ref{prop:vflp}; the proofs are almost identical, but we use
Prop.~\ref{prp:CAFSP} for the error bounds.


\begin{proposition} \label{prp:implicitexpSig}
Let $\xi:[A,B] \to \R^+$ be a (not necessarily monotone) convex
function. Let $K_1,K_2 \geq 1$ and $\Sig>0$. Let $\psi:[A,B] \to \R^+$ be an
increasing piecewise linear convex function with breakpoints $A =
a_1<\ldots<a_n <B$ and slopes $0=\Delta_0 \leq \Delta_1 < \ldots <
\Delta_n$ that $(\Sig,K_1)$-approximates $\xi$. Let $D$ be a (not necessarily
nonnegative) truncated continuous r.v.\ with support $[D^{\min},
  D^{\max}]$ and such that $\min\{\Pr(D = D^{\min}), \Pr(D =
D^{\max})\} > 0$. Suppose $D$ admits a (generalized) PDF $F'$, and let
$F$ be the corresponding CDF. Let $\tilde F(\cdot)$ be a monotone
$K_2$-approximation of $F$ with breakpoints at $D^{\min} = d_1 <
\ldots < d_m = D^{\max}$, and let $\tilde F(d_0)=0$. Let
$f(x,d)=bx+e-d$ for some given $b,e \in \R$, and $m_i(x) = \max \{j
\;|\; d_j \leq bx+e-a_i\}$. Then
\begin{equation}
  \tilde \xi(x) =
\psi(A)+\sum_{i=1}^n (\Delta_i-\Delta_{i-1}) \left(
(bx+e-d_{m_i(x)}-a_i) \tilde F(d_{m_i(x)+1})+ \sum_{k=1}^{m_i(x)-1}
(d_{k+1}-d_k) \tilde F(d_{k+1}) \right)
\label{eq:tildexiSig}
\end{equation}
is a convex piecewise linear $(\Sig K_2,K_1K_2)$-approximation function of
$\mathbb{E}_D(\xi(f(x,D)))$ with $O(mn)$ pieces.  Furthermore, one can
construct in $O(mn\log n)$ time an oracle for $\tilde \xi(\cdot)$ in
the form of either a canonical representation, or a representation
that consists of breakpoints and slopes, each of size $O(mn)$.
\end{proposition}

\begin{proposition} \label{prop:expSig}
Let $\xi:[A,B] \to \R^+$ be a convex function. Let $K_1, K_2 \geq 1$
and $\Sig>0$. Let $\psi:[A,B] \to \R^+$ be a piecewise linear convex
function with $n$ breakpoints that $(\Sig, K_1)$-approximates $\xi$
over $[A,B]$. Let $D$ be a (not necessarily nonnegative) continuous
r.v.\ satisfying the conditions stated in
Prop.~\ref{prp:implicitexp}. Let $\tilde F(\cdot)$ be a monotone
nondecreasing $K_2$-approximation of the CDF of $D$ with $m$
breakpoints. Let $f(x,D) = b x + e + cd$, with $b, e, c \in \R$. Then
we can construct in $O(mn\log n)$ time a convex piecewise linear $(K_2
\Sig, K_1K_2)$-approximation function of $\mathbb{E}_D(\xi(f(x,D)))$
in the form of either a canonical representation, or a representation
that consists of breakpoints and slopes, each of size $O(mn)$.
\end{proposition}

\begin{proposition}
  \label{prop:vflpSig}
  Suppose the DP formulation \eqref{eq:dp} satisfies Conditions
  \ref{con:sets}-\ref{con:functions}. Let $K_1, K_2, K_3 > 1, \;
  \Sig>0$. Let $\hat{z}_{t+1}$ be a piecewise linear convex
  $K_1$-approximation of the value function $z_{t+1}$. For a given
  $I_t \in S_t$, let $\tilde{Z}_{t+1}(\cdot)$ be a piecewise linear
  convex $K_2$-approximation of $\mathbb{E}[\hat{z}_{t+1}(\cdot +
    \theta^I I_t + \vec{\theta}^D \cdot \vD_t)]$, and
  $\tilde{G}_t^D(\cdot)$ be a piecewise linear convex $(\Sig,
  K_3)$-approximation of $\mathbb{E}[g_t^D(\cdot + \sigma^I I_t +
    \vec{\sigma}^D \cdot \vD_t)]$. Then the value of the following
  mathematical program:
  \begin{equation}
    \left.
    \begin{array}{rrcl}
      \displaystyle \bar{z}_t(I_t) := \min_{\vx_t} & g_t^I(I_t, \vx_t) +
      \tilde{G}_t^D(w) + \tilde{Z}_{t+1}(y) & & \\
      & A_t \vx_t &\ge& \vec{b}_t + \vec{\delta}_{b_t} I_t \\
      & \vec{\sigma}^x \cdot \vx_t - w &=& 0 \\
      & \vec{\theta}^x \cdot \vx_t - y &=& 0 \\
      &\vx_t &\ge& 0,
    \end{array}
    \right\}
    \label{eq:vflpSig}
  \end{equation}
  is a $(\Sig, \max \{K_1K_2, K_3\})$-approximation of the optimal value
  function $z_t(I_t)$. \eqref{eq:vflpSig} can be cast as an LP using the
  standard $\min\max$ reformulation of piecewise linear convex
  functions.
\end{proposition}

\begin{algorithm}[tb!]
  \caption{Procedure {\sc ApxScheme2}$(\epsilon)$ for Condition~\ref{con:functions}(ii)}
  \label{alg:fptas2}
  \begin{algorithmic}[1]
    \small
    \STATE $\bar{\epsilon} \leftarrow \frac{\epsilon}{2T+1}, K \leftarrow 1 + \bar{\epsilon}, \bar{K} \leftarrow 1 + \frac{\bar{\epsilon}}{2}$
    \STATE Using Prop.~\ref{prp:directconvSP}, let $\hat z_{T+1}$ be a $(\frac{g^{\min}_{T+1}\bar{\epsilon}}{2}, \bar{K} )$-approximation of $g_{T+1}$
    \FOR{$t:=T$ {\bf downto} 1}
    \STATE Compute  a $(\frac{g^{\min}_{T+1}\bar{\epsilon}}{2}, K)$-approximation of $g^D_{t}$, called $\tilde g^D_{t}$, using Prop.~\ref{prp:directconvSP}
    \STATE \label{step2:convolution1} $\tilde{F}_g \leftarrow$ {\sc CompressConvolution}$(\vD_t, \vec{\sigma}^D, K)$
    \STATE \label{step2:convolution2} $\tilde{F}_z \leftarrow$ {\sc CompressConvolution}$(\vD_t, \vec{\theta}^D, \bar{K})$
    \STATE \label{step2:gtilde} For fixed $I_t$, define $\tilde{G}_t^D \leftarrow$ {\sc CompressExpVal}$(\tilde{g}_t^D, (1, \sigma^I I_t, 1) , \tilde{F}_g)$ \\ /* $\tilde G_t^D(\cdot)$ is an oracle for a $(\frac{g^{\min}_{T+1}\bar{\epsilon}K}{2}, K^2)$-approximation of $\mathbb{E}[g^D_t(\cdot+\sigma^I I_t+\vec{\sigma}^D \cdot \vD_t)]$ */
    \STATE \label{step2:ztilde} For fixed $I_t$, define $ \tilde{Z}_{t+1} \leftarrow$ {\sc CompressExpVal}$(\hat{z}_{t+1}, (1, \theta^I I_t, 1), \tilde{F}_z)$ \\ /* $\tilde Z_{t+1}(\cdot)$ is an oracle for a $\bar{K}$-approximation of $\mathbb{E}[\hat{z}_{t+1}(\cdot+\theta^I I_t+\vec{\theta}^D \cdot \vD_t)]$ */
    \STATE \label{step2:zhat} $\hat{z}_t \leftarrow$ {\sc ScaledCompressConv}$(\bar{z}_t, [\S_{t}^{\min}, \S_{t}^{\max}], K)$, where $\bar{z}_t$ is defined as in \eqref{eq:vflpSig}
    \ENDFOR \RETURN $\hat{z}_1$
  \end{algorithmic}
\end{algorithm}

Equipped with these concepts, we analyze the FPTAS for the case of
Condition \ref{con:functions}(ii).

\begin{theorem} \label{thm:FPTAS2}
  Given a stochastic DP satisfying Conditions
  \ref{con:sets}-\ref{con:functions}(ii), {\sc ApxScheme2}($\epsilon$)
  (Algorithm~\ref{alg:fptas2}) computes a $(1+\epsilon)$-approximation
  of the optimal value function $z_1$, and runs in polynomial time in
  the binary input size and $1/\epsilon$.
\end{theorem}
\begin{proof}
  We outline here only the differences from the proof of
  Thm.~\ref{thm:FPTAS1}, focusing on the computation of the
  approximation factor. The runtime analysis is very similar to that
  of Thm.~\ref{thm:FPTAS1}, with one main modification: $g^D_t$ is not
  described by $m_t$ breakpoints and slopes, but by an approximation
  set of size $O\left(\frac{T}{\epsilon} \log
  \frac{U_g}{g^{\min}_{T+1}\epsilon}\right)$ computed in the course of
  the algorithm. This has repercussions through the analysis; however,
  it is easy to verify that it yields only a polynomial difference
  with respect to the runtime of Thm.~\ref{thm:FPTAS1}.

  The proof proceeds by induction for $t=T+1,\dots,1$, showing that at
  stage $t$ we obtain a piecewise linear convex
  $K^{2(T+1-t)+1}$-approximation $\hat{z}_t$ of the value function
  $z_t$ via an approximation set of cardinality $O(\log_K
  \frac{(T+2-t)U_g}{g^{\min}_{T+1}\epsilon}) = O(\frac{T}{\epsilon}
  \log \frac{(T+2-t)U_g}{g^{\min}_{T+1}\epsilon})$. The last equality
  is due to the fact that $K = 1 + \frac{\epsilon}{(2T + 1)}$. The
  error bound $K^{2T+1} \le 1 + \epsilon$ is a consequence of
  the inequality $\left(1 + \frac{x}{n}\right)^n \le 1 + 2x$, which
  holds for every $0 \le x \le 1$ and $n \in \N$.

  The step $t=T+1$ is trivial, because $\hat{z}_{T+1}$ is a
  $(\frac{g^{\min}_{T+1}\bar{\epsilon}}{2}, \bar{K})$-approximation of
  $g_{T+1}$. By definition of $(\Sig,\Pi)$-approximation functions
  $\hat{z}_{T+1}(I_{T+1}) \geq z_{T+1}(I_{T+1})$ for all
  $I_{T+1}$. Furthermore,
  \begin{align*}
  \hat{z}_{T+1}(I_{T+1}) &\leq \frac{g^{\min}_{T+1}\bar{\epsilon}}{2} +
  \left(1 + \frac{\bar{\epsilon}}{2}\right) z_{T+1}(I_{T+1}) =
  \frac{g^{\min}_{T+1}\epsilon}{2(2T+1)} + \left(1 +
  \frac{\epsilon}{2(2T+1)}\right) z_{T+1}(I_{T+1}) \\ & \leq
  \left(1 + \frac{\epsilon}{2T+1}\right)z_{T+1}(I_{T+1}) \leq K
  z_{T+1}(I_{T+1}),
  \end{align*}
  where the second inequality follows from the fact that $\tilde
  g^{\min}_{T+1}$ is a lower bound for $g_{T+1}$ and consequently for
  $z_{T+1}$. By Prop.~\ref{prp:directconvSP}, we obtain a piecewise
  linear convex explicit description of $\hat g_{T+1}$ via an
  approximation set of the stated cardinality.  We now show the
  induction step.


  At step \ref{step2:gtilde}, for a fixed value of $I_t$ we define a
  $(\frac{g^{\min}_{T+1}\bar{\epsilon}K}{2}, K^2)$-approximation
  $\tilde{G}_t^D(\cdot)$ of $\mathbb{E}_{\vD_t}[g_t^D(\cdot + \sigma^I
    I_t + \vec{\sigma}^D \cdot \vD_t)]$, using Prop.~\ref{prop:expSig}
  with parameters set to $\xi=g_t^D, \; \psi=\tilde g_t^D, \;
  n=O\left(\frac{T}{\eps} \log \frac{U_g}{g^{\min}_{T+1}\epsilon}
  \right), \; m=O(\frac{T \ell}{\epsilon} \log
  \frac{1}{\gamma}), \; K_1= K,
  \Sigma=\frac{g^{\min}_{T+1}\bar{\epsilon}}{2}, K_2= K$.


  At step \ref{step2:ztilde}, for a fixed value of $I_t$ we define a
  $\bar{K}$-approximation $\tilde{Z}_{t+1}(\cdot)$ of
  $\mathbb{E}[\hat z_{t+1}(\cdot + \theta^I I_t + \vec{\theta}^D \cdot
    \vD_t)]$, using Prop.~\ref{prop:exp} with parameters set to
  $\xi=\psi=\hat z_{t+1}, \; n=O(\frac{T}{\epsilon} \log
  \frac{(T+1-t)U_g}{g^{\min}_{T+1}\epsilon}), \; m=O(\frac{T\ell}{\epsilon}
  \log \frac{1}{\gamma}), \; K_1=1$ and $K_2=\bar{K}$. Given a
  value for $I_t$, we compute $\tilde{Z}_{t+1}$ in the form of a
  canonical representation of size $O(\frac{T^2 \ell}{\epsilon^2} \log
  \frac{(T+1-t)U_g}{g^{\min}_{T+1}\epsilon} \log\frac{1}{\gamma})$.

  At step \ref{step2:zhat} we compute an approximation $\hat z_t$ for
  the value function $z_t$, using the function $\bar{z}_t$ as defined
  in \eqref{eq:vflpSig}. By Prop.~\ref{prop:vflpSig} with parameters
  set to $K_1=K^{2(T-t)+1}$, $K_2=\bar{K}$, $K_3=K^2$, and $\Sigma =
  \frac{g^{\min}_{T+1}\bar{\epsilon}K}{2}$, we have that $\bar{z}_t$
  is a $(\frac{g^{\min}_{T+1}\bar{\epsilon}K}{2},
  K^{2(T-t)+1}\bar{K})$-approximation of $z_t$. This is also a
  $K^{2(T+1-t)}$-approximation of it, because:
  \begin{align*}
    \hat{z}_{t}(I_t) &\leq \frac{g^{\min}_{T+1}\bar{\epsilon}K}{2} +
    K^{2(T-t)+1}\left(1 + \frac{\bar{\epsilon}}{2}\right) z_{t}(I_t) \leq
    K^{2(T-t)+1}\left(\frac{\bar{\epsilon}}{2K^{2(T-t)}} + 1 + \frac{\bar{\epsilon}}{2}\right)z_{t}(I_t) \\ &\leq
    K^{2(T-t)+1}\left(1 + \bar{\epsilon}\right)z_{t}(I_t) = K^{2(T-t)+2} z_{t}(I_t),
  \end{align*}
  where we used the facts that $g^{\min}_{T+1} \le z_t(I_t)$ for all
  $I_t$, and $K^{2(T-t)} \ge 1$ since $t \le T$.  The rest of the
  proof is identical to Thm.~\ref{thm:FPTAS1}, updating the runtime as
  described at the beginning of this proof.

\end{proof}

\section{Concluding remarks}
\label{s:conclusions}
This paper presents an FPTAS for stochastic DPs with continuous scalar
state spaces and polyhedral action spaces. To construct an
approximation algorithm, we introduce several tools within the
framework of $K$-approximation sets and functions. More specifically,
we show how to approximately compute the convolution of a finite
number of continuous random variables, the expectation of a convex
function applied to a linear transformation of a continuous random
variable, and how to bound the size of the numbers of a
$K$-approximation set. Combining these tools, we obtain an FPTAS for a
general class of DP models. These models can be seen as multistage
stochastic LPs with one variable linking the stages. The most
important open question is whether our results can be extended to the
case of two (or more) variables linking the stages, or, in other
words, two-dimensional state spaces for the DP. Under an oracle model
for the cost functions, this paper shows that the problem is
intractable even for the two-dimensional case. If the cost functions
are known explicitly, the tools developed here do not suffice to
settle the approximability status of the problem. This is left for
future research.



\bibliographystyle{siamplain}
\bibliography{fptas,apx}

\clearpage

\appendix

\section*{Appendix}

\section{Additional routines}
\label{apx:functions}
We report here the pseudocode and running time of the routines taken
from \cite{halman14full,nannifptascdpfull} referenced in the main text
of the paper. Given a monotone nondecreasing function $\varphi$, we
define a routine {\sc FuncSearchInc}$(\varphi,D,\ell,u)$ that looks
for a point $x \in D$ such that $\ell \le \varphi(x) \le
u$. Implementing this function is straightforward for both discrete
domains and real intervals, using binary search.

In Alg.~\ref{alg:compressinc} we formally define the routine {\sc
  CompressInc} that constructs an oracle for a $K$-approximation function of a monotone nondecreasing $\vp$ in the form of a canonical
representation.

\begin{algorithm}[b!]
  \caption{Function {\sc CompressInc}($\varphi, [A,B], K$).}
  \begin{algorithmic}[1]
    \small
    \STATE $W \leftarrow$ {\sc ApxSetInc}$(\varphi, [A,B], K)$
    \RETURN $\{(x,\varphi(x)) \; | \; x \in W\}$ as an array of tuples sorted by their first coordinate
  \end{algorithmic}
  \label{alg:compressinc}
\end{algorithm}

The main routine used in Algorithm \ref{alg:compressinc} is {\sc
  ApxSetInc}. We give its pseudocode in Algorithm \ref{alg:apxsetinc}
for continuous functions that are bounded away from zero; its
counterpart for functions over discrete domains is straightforward. It
is shown in \cite{nannifptascdpfull} that {\sc ApxSetInc} determines a
$K$-approximation set with $O\left(\frac{1}{\epsilon} \log
\frac{\epsilon \varphi^{\max}}{\varphi^{\min}}\right)$ points in
$O\left((1+t_{\varphi}) (\frac{1}{\epsilon}\log \frac{\epsilon
  \varphi^{\max}}{\varphi^{\min}}) \log ((B-A)\kappa)\right)$
time. For a function defined over a discrete domain $D$, the running
time of {\sc ApxSetInc} becomes $O\left((1+t_{\varphi})
(\frac{1}{\epsilon}\log \frac{\epsilon
  \varphi^{\max}}{\varphi^{\min}}) \log |D|\right)$, see
\cite{halman14full}.

\begin{algorithm}[b!]
  \caption{{\sc Function ApxSetInc}($\varphi, [A, B], K$)}
  \begin{algorithmic}[1]
    \small
    \STATE $x \leftarrow A$, $W \leftarrow \{A, B\}$
    \WHILE{$x < B$}
      \STATE $x \leftarrow$ {\sc FuncSearchInc}$(\varphi,[x,B],\frac{K+1}{2}\varphi(x),
      K\varphi(x))$
      \STATE $W \leftarrow W \cup \{x\}$
    \ENDWHILE
    \RETURN $W$
  \end{algorithmic}
  \label{alg:apxsetinc}
\end{algorithm}

\section{Proof of Prop.~5.2}
\label{apx:proofimplicitexp}
\begin{proof}
For $i=0,\dots,n$, define:
$$\psi_i(y) := \begin{cases} (\Delta_i-\Delta_{i-1})(y-a_i) &
  \textrm{if } y \geq a_i, \\ 0 & \textrm{otherwise}. \end{cases}$$
Because $\psi$ is piecewise linear, for all $y \in [A, B]$ we have
$\psi(y)=\psi(A)+\sum_{i=1}^n \psi_i(y)$. We have:
\begin{equation}
\label{eq:expvalxi}
\begin{array}{ll}
  \mathbb{E}_D(\xi(f(x,D))) & \leq \mathbb{E}_D(\psi(f(x,D))) \\
& = \mathlarger{\int}_{d_1}^{d_m} \psi(f(x,d)) F'(d) \diff d\\
    & = \psi(A)+ \mathlarger{\int}_{d_1}^{d_m} \sum_{i=1}^n \psi_i(f(x,d)) F'(d) \diff d\\
    & = \psi(A)+\sum_{i=1}^n \big( \mathlarger{\int}_{d_1}^{d_m} \psi_i(f(x,d))F'(d) \diff d \big) \\
    & = \psi(A)+\sum_{i=1}^n (\Delta_i-\Delta_{i-1}) \mathlarger{\int}_{d_1}^{d_m} \max\{0, (bx+e-d -a_i)\} F'(d) \diff d \\
    & = \psi(A)+\sum_{i=1}^n (\Delta_i-\Delta_{i-1}) \mathlarger{\int}_{d_1}^{\max \{d : bx+e-d \geq a_i\}} (bx+e-d -a_i)  F'(d) \diff d,
\end{array}
\end{equation}
where the first inequality is due to $\psi$ being a
$K_1$-approximation function of $\xi$, and the rest are algebraic
manipulations.

We construct a discrete r.v.\ $\hat{D}$ that takes
values $d_1,\dots,d_{m-1}$ with:
\begin{equation*}
\Pr(\hat{D} = d_j) = \begin{cases} F(d_2) & \textrm{if } j=1 \\
  F(d_{j+1}) - F(d_{j}) & \textrm{if } j=2,\dots,m-1.
\end{cases}
\end{equation*}
It follows that the CDF $\hat{F}$ of $\hat{D}$ is $\hat{F}(d) = \max
\{F(d_{j+1}) : d_j \le d, j=1,\dots,m-1\}$.  In order to compute
expected values of continuous functions of $\hat{D}$ using the
classical integration approach, we define the generalized PDF of
$\hat{D}$ as follows:
\begin{equation*}
  \hat{F}'(d) := \delta(d-d_1)F(d_2) + \sum_{j=2}^{m-1} \delta(d-d_j) (F(d_{j+1})-F(d_{j})).
\end{equation*}
Notice that $\hat{D} \preceq D$ in the usual stochastic order,
because $\Pr(\hat{D} > d) \le \Pr(D > d)$ for all $d$. Since $\hat{D}
\preceq D$ and $(bx+e-d -a_i)$ is a decreasing function in $d$,
it follows that:
\begin{align*}
\mathlarger{\int}_{d_1}^{\max \{d : bx+e-d \geq a_i\}} (bx+e-d -a_i)
F'(d) \diff d \le \mathlarger{\int}_{d_1}^{\max \{d : bx+e-d \geq
  a_i\}} (bx+e-d -a_i) \hat{F}'(d) \diff d  = \\
(bx+e-d_1-a_i)F(d_2) + \sum_{j=2}^{m_i(x)} (bx+e-d_j -a_i) (F(d_{j+1}) - F(d_j)) = \\
(bx+e-d_1-a_i)\tilde{F}(d_2) + \sum_{j=2}^{m_i(x)} (bx+e-d_j -a_i) (\tilde{F}(d_{j+1}) - \tilde{F}(d_j)) = \\
(bx+e-d_{m_i(x)}-a_i) \tilde{F}(d_{m_i(x)+1}) + \sum_{k=1}^{m_i(x)-1}(d_{k+1} - d_k) \tilde{F}(d_{k+1}).
\end{align*}
Putting everything together in \eqref{eq:expvalxi}, we obtain:
\begin{equation*}
  \mathbb{E}_D(\xi(f(x,D))) \le \psi(A)+\sum_{i=1}^n (\Delta_i-\Delta_{i-1}) \left((bx+e-d_{m_i(x)}-a_i) \tilde{F}(d_{m_i(x)+1}) + \sum_{k=1}^{m_i(x)-1}(d_{k+1} - d_k) \tilde{F}(d_{k+1})\right).
\end{equation*}
Because $\psi(x) \le K_1\xi(x)$ for all $x$ and $F'$ is nonnegative, we
can write:
\begin{equation}%
\label{eq:psixipart1}
\begin{array}{rl}
\displaystyle \mathbb{E}_D(\psi(f(x,D))) &= \mathlarger{\int}_{d_1}^{d_m} \psi(f(x,d)) F'(d) \diff d\\
\displaystyle&\leq \displaystyle K_1 \mathlarger{\int}_{d_1}^{d_m} \xi(f(x,d)) F'(d) \diff d = K_1
\mathbb{E}_D(\xi(f(x,D))).
\end{array}
\end{equation}

We now construct a discrete r.v.\ $\breve{D}$ that takes
values $d_1,\dots,d_m$ with $\Pr(\breve{D} = d_j) = F(d_j) -
F(d_{j-1}) = \tilde{F}(d_j) - \tilde{F}(d_{j-1})$. It follows that the
CDF $\breve{F}$ of $\breve{D}$ is $\breve{F}(d) = \max \{F(d_j) : d_j
\le d, j=1,\dots,m\}$. As before, we define the generalized PDF of
$\breve{D}$ as follows:
\begin{equation*}
  \breve{F}'(d) := \sum_{j=1}^m \delta(d-d_j) (F(d_j)-F(d_{j-1})).
\end{equation*}
Notice that $D \preceq \breve{D}$ in the usual stochastic order,
because $\Pr(D > d) \le \Pr(\breve{D} > d)$ for all $d$. Since $D
\preceq \breve{D}$ and $(bx+e-d -a_i)$ is a decreasing function in
$d$, it follows that:
\begin{eqnarray*}
\mathlarger{\int}_{d_1}^{\max \{d : bx+e-d \geq a_i\}} (bx+e-d -a_i)
F'(d) \diff d \ge \mathlarger{\int}_{d_1}^{\max \{d : bx+e-d \geq
  a_i\}} (bx+e-d -a_i) \breve{F}'(d) \diff d  = \\
\sum_{j=1}^{m_i(x)} (bx+e-d_j -a_i) (F(d_{j}) - F(d_{j-1})) = \sum_{j=1}^{m_i(x)} (bx+e-d_j -a_i) (\tilde{F}(d_{j}) - \tilde{F}(d_{j-1})) = \\
(bx+e-d_{m_i(x)}-a_i) \tilde{F}(d_{m_i(x)}) + \sum_{k=1}^{m_i(x)-1}(d_{k+1} - d_k) \tilde{F}(d_{k}) \ge \\
\frac{1}{K_2}(bx+e-d_{m_i(x)}-a_i) \tilde{F}(d_{m_i(x)+1}) + \frac{1}{K_2}\sum_{k=1}^{m_i(x)-1}(d_{k+1} - d_k) \tilde{F}(d_{k+1}),
\end{eqnarray*}
where the last inequality follows from the fact that $F(d_{j+1}) \le
K_2 F(d_j)$ for $j=1,\dots,m-1$ by definition of $K$-approximation set
for monotone function. Then we can write:
\begin{align}
  \mathbb{E}_D(\psi(f(x,D))) = \psi(A)+\sum_{i=1}^n (\Delta_i-\Delta_{i-1}) \mathlarger{\int}_{d_1}^{\max \{d : bx+e-d \geq a_i\}} (bx+e-d -a_i)  F'(d) \diff d \ge  \notag \\
  \psi(A) +  \notag \\ \frac{1}{K_2}\sum_{i=1}^n (\Delta_i-\Delta_{i-1}) \left( (bx+e-d_{m_i(x)}-a_i) \tilde{F}(d_{m_i(x)+1}) + \sum_{k=1}^{m_i(x)-1}(d_{k+1} - d_k) \tilde{F}(d_{k+1}) \right) \ge \notag \\
 \frac{\tilde \xi(x)}{K_2}. \label{eq:psixipart2}
\end{align}
By combining inequalities \eqref{eq:psixipart1} and
\eqref{eq:psixipart2} we get the desired approximation ratio.

It is easy to verify that $\tilde \xi$ has increasing slopes and is
therefore a convex piecewise linear increasing function. To conclude,
we discuss how to compute a representation of $\tilde \xi$ in terms of
breakpoints and slopes. By examining the expression:
\begin{equation*}
  \tilde \xi(x) =
  \psi(A)+\sum_{i=1}^n (\Delta_i-\Delta_{i-1}) \left(
  (bx+e-d_{m_i(x)}-a_i) \tilde F(d_{m_i(x)+1})+
  \sum_{k=1}^{m_i(x)-1} (d_{k+1}-d_k) \tilde F(d_{k+1}) \right),
\end{equation*}
we see that the slope of each term of the summation changes whenever
$m_i(x)$ changes. There are at most $m$ such changes for each term,
and their location can be computed in $O(m)$ time because for term $i$
the breakpoints are of the form $d_j - e + a_i$ for $j=1,\dots,m$. We
then obtain, in $O(mn)$ time, $n$ sorted lists with $m$ elements
each. These lists can be merged in $O(mn\log n)$ time, yielding a
superset of the breakpoints of $\xi(x)$. To compute the slopes we only
need an additional $O(m)$ time to preprocess the $m-1$ partial sums
$\sum_{k=1}^{m-1}(d_{k+1}-d_k) \tilde F(d_{k+1})$, since the value of
$\tilde F$ at all queried points is known and available in the
approximation set that induces $\tilde F$. The overall time
requirement is therefore $O(mn \log n)$.
\end{proof}

\section{$(\Sigma, \Pi)$-approximation functions and their calculus}
\label{apx:sigpi}
These results are taken from \cite{nannifptascdpfull}.
\begin{proposition} [Adapted from Prop.~3.7 in \cite{nannifptascdpfull}]
 \label{prp:directconvSP}
Let $\vp:[A,B] \ra \R^{+}$ be a $\kappa$-Lipschitz continuous convex
function. Then, for every constants $\Sig > 0$ and $\Pi=1+\eps > 1$, one
can construct a piecewise-linear convex $(\Sig,\Pi)$-approximation
function $\breve \vp$ for $\vp$ with $p:=O\left(\frac{1}{\eps} \log
\frac{\eps \varphi^{\max}}{\Sig}\right)$ pieces in \break
$O\left((1+t_{\varphi}) (\frac{1}{\eps} \log \frac{\eps
  \varphi^{\max}}{\Sig})\log \frac{\kappa(B-A)}{\Sig}\right)$ time,
with explicitly computed breakpoints and slopes. Moreover, the value
of $\breve{\varphi}(\cdot)$ can be determined in $\log p$ time at any
point in $[A,B]$.
\end{proposition}

\begin{proposition}[{\bf Calculus of $(\Sigma,\Pi)$-approximation Functions}]
  \label{prp:CAFSP} For $i=1,2$ let $\Sigma_i \geq 0, \; \Pi_i
  \geq 1$, let $\varphi_i:D \ra \R^+$ be an arbitrary function over
  continuous domain $D$, and let $\tilde{{\varphi}_i}:D \ra \R^+$ be a
  $(\Sigma_i,\Pi_i)$-approximation of $\varphi_i$. Let $\psi_1:D \ra
  D$, and let $\alpha_i \in \R^+$. The following rules hold:
\begin{enumerate}[noitemsep]
\item $\varphi_1$ is a (0,1)-approximation of itself,
\item \label{item:CAFSPlin}(linearity of appr.)
  $ \alpha_1  \tilde{\varphi_1} + \alpha_2 $ is a $(\alpha_1  \Sigma_1,\Pi_1)$-approximation of $\alpha_1  \varphi_1 + \alpha_2 $,
\item \label{item:CAFSPsum}(summation of appr.)
  $\tilde{\varphi_1}+\tilde{\varphi_2}$ is a $(\Sigma_1+\Sigma_2,\max\{\Pi_1,\Pi_2\})$-approximation of $\varphi_1 + \varphi_2$,
\item \label{item:CAFSPtra} (composition of appr.) $\tilde{\varphi_1}(\psi_1)$ is a $(\Sigma_1,\Pi_1)$-approximation of $\varphi_1(\psi_1)$,
\item \label{item:CAFSPmin}(minimization of appr.)
  $\min \{\tilde{\varphi_1},\tilde{\varphi_2}\}$ is a $(\max\{\Sigma_1,\Sigma_2\},\max\{\Pi_1,\Pi_2\})$-approximation of $\min\{\varphi_1,\varphi_2\}$,
\item \label{item:CAFSPmax}(maximization of appr.)
  $\max \{\tilde{\varphi_1},\tilde{\varphi_2}\}$ is a $(\max\{\Sigma_1,\Sigma_2\},\max\{\Pi_1,\Pi_2\})$-approximation of $\max\{\varphi_1,\varphi_2\}$,
\item \label{item:CAFSPapx} (approximation of appr.) If $\varphi_2=\tilde{\varphi_1}$ then
  $\tilde{\varphi_2}$ is a $(\Sigma_2+\Pi_2\Sigma_1,\Pi_1\Pi_2)$-approximation of $\varphi_1$.
\end{enumerate}
\end{proposition}

\end{document}